\documentclass[a4paper,USenglish]{lipics-v2019}

\usepackage{amsmath,amssymb,xspace}
\usepackage{graphicx}
\usepackage{url}
\usepackage{multirow}
\usepackage{wrapfig}
\usepackage[nodisplayskipstretch]{setspace} 
\usepackage{comment}
\usepackage[boxruled,vlined,linesnumbered]{algorithm2e}
\usepackage{tikz,tikz-qtree}
\usetikzlibrary{matrix,fit}
\usetikzlibrary{arrows,decorations.pathreplacing,calc}
\usepackage{graphicx}
\usepackage[nomessages]{fp}

\usepackage{lipsum}

\usepackage{listings,xcolor}

\usepackage{breakurl}
\usepackage{pdfpages}
\usepackage{float,xspace}
\usepackage{hyperref}

\usepackage{microtype}

\usepackage{marginnote}

\nolinenumbers

\newcommand{\NMS}{\textsc{NaturalMergeSort}\xspace}
\newcommand{\MS}{\textsc{MergeSort}\xspace}
\newcommand{\MinS}{\textsc{MinimalSort}\xspace}
\newcommand{\TS}{\textsc{TimSort}\xspace}

\newcommand{\QS}{\textsc{QuickSort}\xspace}

\newcommand{\caseX}{$\#$\textsc{X}\xspace}

\newcommand{\Python}{{Python}\xspace}
\newcommand{\Java}{{Java}\xspace}

\renewcommand{\H}{\mathcal{H}}
\renewcommand{\O}{\mathcal{O}}
\renewcommand{\S}{\mathcal{S}}

\newcommand{\C}{\mathcal{C}}
\newcommand{\R}{\mathcal{R}}
\newcommand{\rs}{\mathsf{r}}

\newcommand{\cost}{\textbf{c}}

\DeclareMathOperator{\height}{\texttt{height}}
\DeclareMathOperator{\rundecomp}{\texttt{runs}}
\DeclareMathOperator{\runstack}{\R}

\newcommand{\true}{\textbf{true}}
\newcommand{\pot}{\mathsf{pot}}

\newcommand{\ctok}{$\diamondsuit$\xspace}
\newcommand{\stok}{$\heartsuit$\xspace}

\renewcommand{\gets}{\ensuremath{\leftarrow}}
\newenvironment{disjunction}{\begin{itemize}}{\vspace{-\baselineskip}\end{itemize}}
\theoremstyle{plain}
\newtheorem{proposition}[theorem]{Proposition}
\newtheorem{claim}[theorem]{Claim}

\theoremstyle{definition}
\newtheorem{remark2}[theorem]{Remark}

\title{On the Worst-Case Complexity of TimSort} 

\author{Nicolas Auger, Vincent Jugé, Cyril Nicaud, and Carine Pivoteau}{Universit\'e Paris-Est, LIGM (UMR 8049), UPEM, F77454 Marne-la-Vall\'ee, France}{}{}{}

\authorrunning{N. Auger, V. Jugé, C. Nicaud, and C. Pivoteau} 

\Copyright{Nicolas Auger, Vincent Jugé, Cyril Nicaud, and Carine Pivoteau}

\subjclass{}


\EventEditors{}
\EventNoEds{1}
\EventLongTitle{}
\EventShortTitle{}
\EventAcronym{}
\EventYear{}
\EventDate{}
\EventLocation{}
\EventLogo{}
\SeriesVolume{}
\ArticleNo{} 
\begin{document}

\maketitle

\begin{abstract}
\TS is an intriguing sorting algorithm designed in 2002 for Python, 
whose worst-case complexity was announced, but not proved until our recent preprint.
In fact, there are two slightly different versions of \TS that are currently implemented in Python and in Java respectively. 
We propose a pedagogical and insightful proof that the Python version runs in time~$\O(n \log n)$. 
The approach we use in the analysis also applies to the Java version, although not without very involved technical details. 
As a byproduct of our study, we uncover a bug in the Java implementation that can cause the sorting method to fail during the execution. 
We also give a proof that Python's \TS running time is in $\O(n + n \H)$,
where $\H$ is the entropy of the distribution of runs (i.e. maximal monotonic sequences), which is quite a natural parameter here and part of the explanation for the good behavior of \TS on partially sorted inputs. Finally, we evaluate precisely the worst-case running time of Python's \TS,
and prove that it is equal to $1.5 n \H + \mathcal{O}(n)$.
\end{abstract}

\section{Introduction}\label{intro}
\TS is a sorting algorithm designed in 2002 by Tim Peters~\cite{Peters2015},
for use in the \Python programming language. It was thereafter implemented in
other well-known programming languages such as \Java. The algorithm includes many implementation optimizations, a few heuristics and some refined tuning, but its high-level principle is rather simple:
The sequence $S$ to be sorted is first decomposed greedily into monotonic runs (i.e. 
nonincreasing or nondecreasing subsequences of~$S$ as depicted on Figure~\ref{fig:runs}), which are then merged pairwise according to some specific rules.

\begin{figure}[h]
\centerline{$
S=(~\underbrace{12,10,7,5}_{\text{first run}},
~\underbrace{7,10,14,25,36}_{\text{second run}},
~\underbrace{3,5,11,14,15,21,22}_{\text{third run}},
~\underbrace{20,15,10,8,5,1}_{\text{fourth run}}~)
$}
\caption{A sequence and its {\em run decomposition} computed by \TS: for each run, the first two elements determine if it is increasing or decreasing, then it continues with the maximum number of consecutive elements that preserves the monotonicity.\label{fig:runs}}
\end{figure}

The idea of starting with a decomposition into runs is not new, and already appears in Knuth's \NMS~\cite{Knuth98}, where increasing runs are sorted using the same mechanism as in \MS. 
Other merging strategies combined with decomposition into runs appear in the literature, such as the \MinS of~\cite{Ta09} (see also~\cite{BaNa13} for other considerations on the same topic). 
All of them have nice properties: they run in $\O(n\log n)$ and even $\O(n+n\log\rho)$, where $\rho$ is the number of runs, which is optimal in the model of sorting by comparisons~\cite{Mannila1985}, using the classical counting argument for lower bounds. 
And yet, among all these merge-based algorithms, \TS was favored in several very popular programming languages, which suggests that it performs quite well in practice. 

\TS running time was implicitly assumed to be $\O(n\log n)$, but our unpublished preprint~\cite{AuNiPi15} contains, to our knowledge, the first proof of it. This was more than ten years after \TS started being used instead of \QS in several major programming languages. 
The growing popularity of this algorithm invites for a careful theoretical investigation. In the present paper, we make a thorough analysis which provides a better understanding of the inherent qualities of the merging strategy of \TS. 
Indeed, it reveals that, even without its refined heuristics,\footnote{These heuristics are useful in practice, but do not improve the worst-case complexity of the algorithm.} this is an effective sorting algorithm, computing and merging runs on the fly, using only local properties to make its decisions. 

We first propose in Section~\ref{sec:analysis1} a new pedagogical and self-contained exposition that \TS runs in time $\O(n + n \log n)$, which we want both clear and insightful.
In fact, we prove a stronger statement: on an input consisting of $\rho$ runs
of respective lengths $r_1,\ldots,r_\rho$, we establish that \TS runs
in $\O(n + n \H) \subseteq \O(n + n \log \rho) \subseteq \O(n + n \log n)$,
where $\H = H(r_1/n,\ldots,r_\rho/n)$ and
$H(p_1,\ldots,p_\rho) = - \sum_{i=1}^\rho p_i \log_2(p_i)$ is the binary Shannon entropy.

We then refine this approach, in Section~\ref{sec:analysis2},
to derive precise bounds on the worst-case running time of \TS,
and we prove that it is equal to $1.5 n \H + \mathcal{O}(n)$.
This answers positively a conjecture of~\cite{BuKno18}.
%
%
%
%
Of course, the first result follows from the second, but since we believe that each one is interesting on its own, we devote one section to each of them.

To introduce our last contribution, we need to look into the evolution of the algorithm: there are actually not one, but two main versions of \TS. The first version of the algorithm contained a flaw, which was spotted in~\cite{GoRoBoBuHa15}: while the input was correctly sorted, the algorithm did not behave as announced (because of a broken invariant).
This was discovered by De Gouw and his co-authors while trying to prove formally the correctness of \TS. 
They proposed a simple way to patch the algorithm, which was quickly adopted in Python, leading to what we consider to be the real \TS. This is the one we analyze in Sections~\ref{sec:analysis1} and~\ref{sec:analysis2}. 
On the contrary, Java developers chose to stick with the first version of \TS, and adjusted some tuning values (which depend on the broken invariant; this is explained in Sections~\ref{presentation} and~\ref{sec:java}) to prevent the bug exposed by~\cite{GoRoBoBuHa15}. Motivated by its use in Java, we explain in Section~\ref{sec:java} how, at the expense of very complicated technical details, the elegant proofs of the Python version can be twisted to prove the same results for this older version. While working on this analysis, we discovered yet another error in the correction made in~Java. Thus, we compute yet another patch, even if we strongly agree that the algorithm proposed and formally proved in~\cite{GoRoBoBuHa15} (the one currently implemented in Python) is a better option.

\section{TimSort core algorithm}\label{presentation}

\begin{algorithm}[t]
\begin{small}
\SetArgSty{texttt}
\DontPrintSemicolon
\SetKwInOut{Input}{Input}
\Input{A sequence $S$ to sort}
\KwResult{The sequence $S$ is sorted into a single run, which  remains on the 
stack.}
\BlankLine
\SetKwInput{KwData}{Note}
\KwData{The function {\tt merge\_force\_collapse} repeatedly pops the last two runs on the stack~$\runstack$, merges them and pushes the resulting run back on the stack.}
\BlankLine
$\rundecomp \gets $ a run decomposition of $S$\;
$\runstack \gets $ an empty stack\;
\While(\tcp*[f]{main loop of \TS}){$\rundecomp\neq \emptyset$\label{algline:begin_loop}}{
  remove a run $r$ from $\rundecomp$ and push $r$ onto $\runstack$\;
  {\tt merge\_collapse}($\runstack$)\;\label{algline:end_loop}
}
\If(\tcp*[f]{the height of $\runstack$ is its number of runs}){$\height(\runstack) \neq 1$}{
  {\tt merge\_force\_collapse}($\runstack$)
}
\end{small}
\caption{\TS \hfill(Python 3.6.5) \label{alg:TimSortMainLoop}}
\end{algorithm}

The idea of \TS is to design a merge sort that can exploit the possible  
``non randomness'' of the data, without having to detect it beforehand and 
without damaging the performances on random-looking data. This follows the ideas 
of adaptive sorting (see~\cite{Mannila1985} for a survey on taking presortedness 
into account when designing and analyzing sorting algorithms).

The first feature of \TS is to work on the natural decomposition of the input  
sequence into maximal runs. In order to get larger subsequences, \TS allows both 
nondecreasing and decreasing runs, unlike most merge sort algorithms. 

Then, the merging strategy of \TS~(Algorithm~\ref{alg:TimSortMainLoop}) is quite simple yet very efficient. The runs are considered in the order  
given by the run decomposition and successively pushed onto a stack. 
If some conditions on the size of the topmost runs of the stack are not satisfied after a new run has been pushed, this can trigger a series of merges between pairs of runs at the top or right under. 
And at the end, when all the runs in the initial decomposition 
have been pushed, the last operation is to merge the remaining runs two by two, starting
at the top of the stack, to get a sorted sequence.
The conditions on the stack and the merging rules are implemented in the  
subroutine called~{\tt merge\_collapse} detailed in Algorithm~\ref{alg:merge_collapse}. 
This is what we consider to be \TS core mechanism and this is the main focus of our analysis. 

\begin{algorithm}[t]
\begin{small}
\SetArgSty{texttt}
\DontPrintSemicolon
\SetKwInOut{Input}{Input}
\Input{A stack of runs $\runstack$}
\KwResult{The invariant of Equations~\eqref{eq:inv1} and~\eqref{eq:inv2} is established.}
\BlankLine
\SetKwInput{KwData}{Note}
\KwData{The runs on the stack are denoted by $\runstack[1]\dots\runstack[\height(\runstack)]$, from top to bottom. The length of run $\runstack$[i] is denoted by $\rs_i$. The blue highlight indicates that the condition was not present in the original version of \TS (this will be discussed in section~\ref{sec:java}).}
\BlankLine
\While{$\height(\runstack)>1$}{
  $n \gets \height(\runstack) - 2$ 
    \BlankLine
  \If{
  ($n > 0$ and $\rs_{3}\leqslant \rs_{2} + \rs_{1}$)
  \textcolor{blue}{ or ($n > 1$ and $\rs_{4} \leqslant \rs_{3} + \rs_{2}$)}\,\label{algline:new_cond}}{
    \If{$\rs_{3} < \rs_{1}$}{
      merge runs $\runstack$[2] and $\runstack$[3] on the stack\;
    }
    \lElse{merge runs $\runstack$[1] and $\runstack$[2] on the stack}
  }
  \ElseIf{$\rs_{2} \leqslant \rs_{1}$}{
    merge runs $\runstack$[1] and $\runstack$[2] on the stack
  }
  \lElse{
    break
  }
}
\end{small}
\caption{The {\tt merge\_collapse} procedure \hfill(Python 3.6.5)\label{alg:merge_collapse}}
\end{algorithm}

Another strength of \TS is the use of many effective heuristics to save
time, such as ensuring that the initial runs are not to small thanks to an insertion sort 
or using a special technique called ``galloping'' to optimize the merges. 
However, this does not interfere with our analysis and we will not discuss this matter any further.

Let us have a closer look at Algorithm~\ref{alg:merge_collapse} which is a pseudo-code transcription
of the {\tt merge\_collapse} procedure found in the latest version of Python (3.6.5). 
To illustrate its mechanism, an example of execution of the main loop of \TS (lines~\ref*{algline:begin_loop}-\ref*{algline:end_loop} of Algorithm~\ref{alg:TimSortMainLoop}) is given in Figure~\ref{fig:ts-python-exec}. 
As stated in its note~\cite{Peters2015}, Tim Peter's idea was that: 
\begin{quote}
``The thrust of these rules when they trigger merging is to balance the run
lengths as closely as possible, while keeping a low bound on the number of
runs we have to remember.''    
\end{quote}
To achieve this, the merging conditions of {\tt merge\_collapse} are designed  
to ensure that the following invariant\footnote{Actually, in~\cite{Peters2015}, the invariant is only stated for the 3 topmost runs of the stack.} is true at the end of the procedure:
\begin{eqnarray}
    \rs_{i+2}&>&\rs_{i+1}+\rs_{i}, \label{eq:inv1}\\
    \rs_{i+1}&>&\rs_{i}.\label{eq:inv2}
\end{eqnarray}
This means that the runs lengths $\rs_i$ 
on the  stack grow at least as fast as the Fibonacci numbers and, therefore, 
that the height of the stack stays logarithmic (see Lemma~\ref{lm:h-is-small}, section~\ref{sec:analysis1}). 

Note that the bound on the height of the stack is not enough to justify the $\O(n\log n)$ running time of \TS. Indeed, without the smart strategy used to merge the runs ``on the fly'', it is easy to build an example using a stack containing at most two runs and that gives a $\Theta(n^2)$ complexity: just assume that all runs have size two, push them one by one onto a stack and perform a merge each time there are two runs in the stack. 

We are now ready to proceed with the analysis of \TS complexity. As mentioned earlier, Algorithm~\ref{alg:merge_collapse} does not correspond to the first implementation of \TS in Python, nor to the current one in Java, but to the latest Python version. The original version will be discussed in details later, in Section~\ref{sec:java}.   

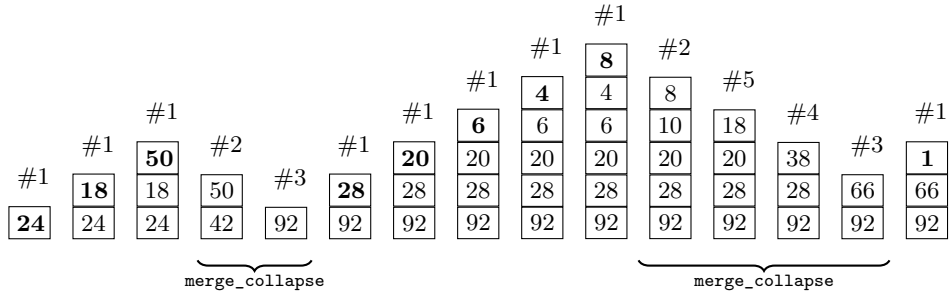
\begin{figure}[t]
\begin{center}
\begin{small}
\setlength{\tabcolsep}{1pt}
\begin{tabular}{ccccccccccccccc}
\begin{tikzpicture}[every node/.style={align=center,text width=1em}]
\matrix (A) [matrix of nodes,nodes={draw}] { {\bf 24}\\ };
\node [above] at (A.north) {$\#1$};
\end{tikzpicture}
&
\begin{tikzpicture}[every node/.style={align=center,text width=1em}]
\matrix (A) [matrix of nodes,nodes={draw}] { {\bf 18}\\24\\ };
\node [above] at (A.north) {$\#1$};
\end{tikzpicture}
&
\begin{tikzpicture}[every node/.style={align=center,text width=1em}]
\matrix (A) [matrix of nodes,nodes={draw}] { {\bf 50}\\18\\24\\ };
\node [above] at (A.north) {$\#1$};
\end{tikzpicture}
&
\begin{tikzpicture}[every node/.style={align=center,text width=1em}]
\matrix (A) [matrix of nodes,nodes={draw}] { 50\\42\\ };
\node [above] at (A.north) {$\#2$};
\end{tikzpicture}
&
\begin{tikzpicture}[every node/.style={align=center,text width=1em}]
\matrix (A) [matrix of nodes,nodes={draw}] { 92\\ };
\node [above] at (A.north) {$\#3$};
\end{tikzpicture}
&
\begin{tikzpicture}[every node/.style={align=center,text width=1em}]
\matrix (A) [matrix of nodes,nodes={draw}] { {\bf 28}\\92\\ };
\node [above] at (A.north) {$\#1$};
\end{tikzpicture}
&
\begin{tikzpicture}[every node/.style={align=center,text width=1em}]
\matrix (A) [matrix of nodes,nodes={draw}] { {\bf 20}\\28\\92\\ };
\node [above] at (A.north) {$\#1$};
\end{tikzpicture}
&
\begin{tikzpicture}[every node/.style={align=center,text width=1em}]
\matrix (A) [matrix of nodes,nodes={draw}] { {\bf 6}\\20\\28\\92\\ };
\node [above] at (A.north) {$\#1$};
\end{tikzpicture}
&
\begin{tikzpicture}[every node/.style={align=center,text width=1em}]
\matrix (A) [matrix of nodes,nodes={draw}] { {\bf 4}\\6\\20\\28\\92\\ };
\node [above] at (A.north) {$\#1$};
\end{tikzpicture}
&
\begin{tikzpicture}[every node/.style={align=center,text width=1em}]
\matrix (A) [matrix of nodes,nodes={draw}] { {\bf 8}\\4\\6\\20\\28\\92\\ };
\node [above] at (A.north) {$\#1$};
\end{tikzpicture}
&
\begin{tikzpicture}[every node/.style={align=center,text width=1em}]
\matrix (A) [matrix of nodes,nodes={draw}] { 8\\10\\20\\28\\92\\ };
\node [above] at (A.north) {$\#2$};
\end{tikzpicture}
&
\begin{tikzpicture}[every node/.style={align=center,text width=1em}]
\matrix (A) [matrix of nodes,nodes={draw}] { 18\\20\\28\\92\\ };
\node [above] at (A.north) {$\#5$};
\end{tikzpicture}
&
\begin{tikzpicture}[every node/.style={align=center,text width=1em}]
\matrix (A) [matrix of nodes,nodes={draw}] { 38\\28\\92\\ };
\node [above] at (A.north) {$\#4$};
\end{tikzpicture}
&
\begin{tikzpicture}[every node/.style={align=center,text width=1em}]
\matrix (A) [matrix of nodes,nodes={draw}] { 66\\92\\ };
\node [above] at (A.north) {$\#3$};
\end{tikzpicture}
&
\begin{tikzpicture}[every node/.style={align=center,text width=1em}]
\matrix (A) [matrix of nodes,nodes={draw}] { {\bf 1}\\66\\92\\ };
\node [above] at (A.north) {$\#1$};
\end{tikzpicture}
\end{tabular}
\begin{tikzpicture}
\draw[decorate,thick,decoration={brace,amplitude=4pt,mirror}] (-3.8,0.2) -- (-2.3,0.2); 
\draw[decorate,thick,decoration={brace,amplitude=4pt,mirror}] (2.0,0.2) -- (5.3,0.2); 
\draw (-3.05,-0.1) node{\scriptsize {\tt merge\_collapse}};
\draw (3.65,-0.1) node{\scriptsize {\tt merge\_collapse}};
\node[text width=\textwidth]{};
\end{tikzpicture}
\end{small}
\end{center}
\vspace{-6mm}
\caption{ The successive states of the stack~$\runstack$ (the values are the  
lengths of the runs) during an execution of the main loop of \TS 
(Algorithm~\ref{alg:TimSortMainLoop}), with the lengths of the runs in 
$\rundecomp$ being $(24, 18, 50, 28, 20, 6, 4, 8, 1)$. 
The label $\#1$ indicates that a run has just been pushed onto the stack. 
The other labels refer to the different merges cases of {\tt merge\_collapse} as translated in Algorithm~\ref{alg:TS translated}.
\label{fig:ts-python-exec}}
\end{figure}

\section{TimSort runs in \texorpdfstring{$\O(n\log n)$}{Lg}}\label{sec:analysis1}

At the first release of \TS~\cite{Peters2015}, a time complexity of $\O(n\log n)$ was 
announced with no element of proof given. It seemed to remain unproved until 
our recent preprint~\cite{AuNiPi15}, where we provide a confirmation of this fact, using 
a proof which is not difficult but a bit tedious. This result was  refined later in~\cite{BuKno18}, where the authors provide lower and upper bounds, including explicit multiplicative constants, for different merge sort algorithms.

Our main concern is to provide an insightful proof of the complexity of \TS, in order to highlight how well designed is the strategy used to choose the order in which the merges are performed.
The present section is more detailed than the following ones as we want it to be
self-contained once \TS has been translated into Algorithm~\ref{alg:TS translated} (see below).

\begin{algorithm}[t]
\begin{small}
\SetArgSty{texttt}
\DontPrintSemicolon
\SetKwInOut{Input}{Input}
\Input{A sequence to $S$ to sort}
\KwResult{The sequence $S$ is sorted into a single run,  which remains on the 
stack.}
\SetKwInput{KwData}{Note}
\KwData{At any time, we denote the height of the stack $\runstack$ by $h$
and its $i$\textsuperscript{th} top-most run (for $1 \leqslant i 
\leqslant h$)  
by $R_i$. The size of this run is denoted by $r_i$.}
\BlankLine
\BlankLine
$\rundecomp \gets $ the run decomposition of $S$\;
$\runstack \gets $ an empty stack\;
\While(\tcp*[f]{main loop of \TS}){$\rundecomp\neq \emptyset$}{
  remove a run $r$ from $\rundecomp$ and push $r$ onto 
$\runstack$\tcp*[r]{\#1$\;$} 
  \While{\true}{\label{algline:inner_while_translated}
    \lIf{$h \geqslant 3$ and $r_1 > r_3$}{
        merge the runs $R_2$ and $R_3$ \tcp*[f]{\#2}
    }
    \lElseIf{$h \geqslant 2$ and $r_1 \geqslant r_2$}{
        merge the runs $R_1$ and $R_2$ \tcp*[f]{\#3}
    }
    \lElseIf{$h \geqslant 3$ and $r_1+r_2 \geqslant r_3$}{
        merge the runs $R_1$ and $R_2$ \tcp*[f]{\#4}
    }
    \textcolor{blue}{\lElseIf{$h \geqslant 4$ and $r_2+r_3 \geqslant r_4$}{ \label{algline:new_cond_translated}
        merge the runs $R_1$ and $R_2$ \tcp*[f]{\#5}
    }}
    \lElse{break}
  }
}
\lWhile{$h \neq 1$}{
  merge the runs $R_1$ and $R_2$
}
\end{small}
\caption{TimSort: translation of Algorithm~\ref{alg:TimSortMainLoop} and Algorithm~\ref{alg:merge_collapse}\label{alg:TS translated}}
\end{algorithm}

As our analysis is about to demonstrate, in terms of worst-case complexity, the good performances of \TS do not rely on the way merges are performed. 
Thus we choose to ignore their many optimizations and consider that merging two runs of lengths~$r$ and~$r'$ requires both $r+r'$ element moves and $r+r'$ element comparisons.
Therefore, to quantify the running time of \TS, we only take into account the number of comparisons performed.

In particular, aiming at computing precise bounds on the running time of \TS, we
follow~\cite{golin1993queue,AuNiPi15,BuKno18,munro2018nearly} and define the
\emph{merge cost} for merging two runs of lengths $r$ and $r'$ as $r + r'$,
i.e., the length of the resulting run. Henceforth, we will identify the time spent for
merging two runs with the merge cost of this merge.

\vfill

\begin{theorem}\label{thm:complexity n + n H}
Let $\C$ be the class of arrays of length $n$, whose run decompositions
consist of $\rho$ monotonic runs of
respective lengths $r_1,\ldots,r_\rho$.
Let $H(p_1,\ldots,p_\rho) = - \sum_{i=1}^\rho p_i \log_2(p_i)$ be the binary Shannon entropy,
and let $\H = H(r_1/n,\ldots,r_\rho/n)$.

The running time of \TS on arrays in $\C$ is
$\mathcal{O}(n + n \H)$. 
\end{theorem}

From this result, we easily deduce the following complexity bound on \TS,
which is less precise but more simple.

\begin{theorem}\label{thm:complexity n log n}
The running time of \TS on arrays of length $n$ that consist of $\rho$ monotonic runs is
$\O(n + n\log \rho)$, and therefore $\O(n \log n)$.
\end{theorem}

\begin{proof}
The function $f : x \mapsto - x \ln(x)$ is concave on the interval $\mathbb{R}_{>0}$ of positive real numbers,
since its second derivative is $f''(x) = - 1/x$. Hence, when $p_1,\ldots,p_\rho$ are
positive real numbers that sum up to one, we have
$H(p_1,\ldots,p_\rho) = {\textstyle\sum_{i=1}^\rho f(p_i)/\ln(2)} \leqslant \rho f(1/\rho)/\ln(2) = \log_2(\rho)$.
In particular, this means that $\H \leqslant \log_2(\rho)$, and therefore that
\TS runs in time $\O(n + n\log \rho)$.
Since $\rho \leqslant n$, it further follows that $\O(n + n\log \rho) \subseteq \O(n + n\log n) = \O(n \log n)$, which completes the proof.
\end{proof}

Before proving Theorem~\ref{thm:complexity n + n H},
we first show that it is optimal up to a multiplicative constant,
by recalling the following variant of a result from~\cite[Theorem~2]{BaNa13}.

\begin{proposition}\label{proposition:optimality}
For every algorithm comparing only pairs of elements,
there exists an array in the class $\C$ whose sorting
requires at least $n \H - 3 n$ element comparisons.
\end{proposition}

\begin{proof}
In the comparison model, at least $\log_2(|\C|)$ element comparisons
are required for sorting all arrays in $\C$.
Hence, we prove below that $\log_2(|\C|) \geqslant n \H - 3 n$.

Let $\pi = (\pi_1,\ldots,\pi_\rho)$ be a partition of the set $\{1,\ldots,n\}$
into $\rho$ subsets of respective sizes $r_1,\ldots,r_\rho$;
we say that $\pi$ is \emph{nice} if $\max \pi_i > \min \pi_{i+1}$ for all $i \leqslant \rho-1$.
Let us denote by $\mathcal{P}$ the set of partitions $\pi$ of $\{1,\ldots,n\}$ such that $|\pi_i| = r_i$
for all $i \leqslant \rho$, and by $\mathcal{N}$ the set of nice partitions.

Let us transform every partition $\pi \in \mathcal{P}$ into a nice partition as follows.
First, by construction of the run decomposition of an array, we know that $r_1,\ldots,r_{\rho-1} \geqslant 2$,
and therefore that $\min \pi_i < \max \pi_i$ for all $i \leqslant \rho-1$.
Then, for all $i \leqslant \rho-1$, if $\max \pi_i < \min \pi_{i+1}$, we exchange the partitions to which belong $\max \pi_i$ and $\min \pi_{i+1}$, i.e., we move $\max \pi_i$ from the set $\pi_i$ to $\pi_{i+1}$,
and $\min \pi_{i+1}$ from $\pi_{i+1}$ to $\pi_i$. Let $\pi^\ast$ be the partition obtained after
these exchanges have been performed.

Observe that $\pi^\ast$ is nice, and that at most $2^{\rho-1}$ partitions $\pi \in \mathcal{P}$
can be transformed into $\pi^\ast$. This proves that $2^{\rho-1}|\mathcal{N}| \geqslant |\mathcal{P}|$.
Let us further identify every nice partition $\pi^\ast$ with an array in $\C$,
which starts with the elements of $\pi^\ast_1$ (listed in increasing order),
then of $\pi^\ast_2, \ldots, \pi^\ast_\rho$.
We thereby define an injective map from $\mathcal{N}$ to $\C$,
which proves that $|\C| \geqslant |\mathcal{N}|$.

Finally, variants of the Stirling formula indicate that
$(k/e)^k \leqslant k! \leqslant e \sqrt{k} (k/e)^k$ for all $k \geqslant 1$.
This proves that
\begin{align*}
\log_2(|\C|) & \geqslant \log_2(|\C|) \geqslant (1 - \rho) + \log_2(|\mathcal{P}|) \\
& \geqslant (1 - \rho) + n \log_2(n) - \rho \log_2(e) - {\textstyle\sum_{i=1}^\rho} (r_i+1/2) \log_2(r_i) \\
& \geqslant n \H + (1 - \rho - \rho \log_2(e)) - 1/2 {\textstyle\sum_{i=1}^\rho} \log_2(r_i).
\end{align*}
By concavity of the function $x \mapsto \log_2(x)$, it follows that
$\textstyle\sum_{i=1}^\rho \log_2(r_i) \leqslant \rho \log_2(n/\rho)$.
One checks easily that the function $x \mapsto x \log_2(n/x)$ takes its maximum value at $x = n/e$,
and since $n \geqslant \rho$, we conclude that
$\log_2(|\C|) \geqslant n \H - (1+\log_2(e)+\log_2(e)/e) n \geqslant n \H - 3 n$.
\end{proof}

We focus now on proving Theorem~\ref{thm:complexity n + n H}.
The first step consists in rewriting Algorithm~\ref{alg:TimSortMainLoop} and 
Algorithm~\ref{alg:merge_collapse} in a form that is easier to deal with. This is done in Algorithm~\ref{alg:TS translated}. 

\begin{claim}
For any input, Algorithms~\ref{alg:TimSortMainLoop} and~\ref{alg:TS translated} perform the same comparisons. 
\end{claim}

\begin{proof}
The only difference is that Algorithm~\ref{alg:merge_collapse} was changed into the \texttt{while} loop of lines~5 to~10 in Algorithm~\ref{alg:TS translated}.
Observing the different cases, it is straightforward to verify that merges involving the same runs take place in the same order in both algorithms.
Indeed, if $r_3 < r_1$, then $r_3 \leqslant r_1 + r_2$, and therefore line 5 is triggered in Algorithm~\ref{alg:merge_collapse}, so that
both algorithms merge the $2$\textsuperscript{nd} and $3$\textsuperscript{rd} runs.
On the contrary, if $r_3 \geqslant r_1$, then both algorithms merge the $1$\textsuperscript{st} and $2$\textsuperscript{nd} runs if and only if
$r_2 \leqslant r_1$ or $r_3 \leqslant r_1 + r_2$ ({\color{blue}or $r_4 \leqslant r_2 + r_3$}).
\end{proof}

\begin{remark2}\label{rem:main-loop}
Proving Theorem~\ref{thm:complexity n log n} only requires
analyzing the \emph{main loop} of the algorithm (lines~3 to 10).
Indeed, computing the run decomposition (line~1) can be done on the fly,
by a greedy algorithm, in time linear in $n$, and the \emph{final loop} (line~11) might be performed in the main loop by adding a fictitious run of length $n+1$ at the end of the decomposition.

In the sequel, for the sake of readability, we also omit
checking that $h$ is large enough to trigger the cases \#2 to \#5.
Once again, such omissions are benign, since
adding fictitious runs of respective lengths $8n$, $4n$, $2n$ and $n$
(in this order) at the beginning of the decomposition
would ensure that $h \geqslant 4$ during the whole loop.
\end{remark2}


We sketch now the main steps of our proof, i.e., the amortized analysis of the main loop.
A first step is to establish the invariant~\eqref{eq:inv1} and~\eqref{eq:inv2},
ensuring an exponential growth of the run lengths within the stack.

Elements of the input array are easily identified by their starting position in the array, so we consider them as well-defined and distinct entities (even if they have the same value). 
The \emph{height} of an element in the stack of runs is the number of runs that are
below it in the stack: the elements belonging to the run~$R_i$ in the stack
$\S = (R_1,\ldots, R_h)$ have height $h-i$, and we recall that the length of the run $R_i$ is denoted
by $r_i$.

\begin{lemma}\label{lm:invariant I}
At any step during the main loop of \TS, we have~$r_i+r_{i+1}<r_{i+2}$ for all $i \in \{3,\ldots,h-2\}$.
\end{lemma}

\begin{proof}
We proceed by induction. The proof consists in verifying that, if the
invariant holds at some point, then it still holds when an update of the stack occurs in one of the five situations labeled \#1 to \#5 in the algorithm.
This can be done by a straightforward case analysis.
We denote by $\overline{\S}=(\overline{R}_1,\ldots, \overline{R}_{\overline{h}})$ the new state of the stack after the update:
\begin{disjunction}
  \item If Case \#1 just occurred, a new run $\overline{R}_1$ was pushed.
  This implies that none of the conditions of Cases \#2 to \#5 hold in $\S$, otherwise merges would have continued.  In particular, we have $r_2+r_3<r_4$.
  As $\overline{r}_i=r_{i-1}$ for all $i\geqslant 2$, and since the invariant holds for $\S$, it also
  holds for $\overline{\S}$.
  
  \item If one of the Cases \#2 to \#5 just occurred, $\overline{r}_i=r_{i+1}$ for all $i\geqslant 3$. Since the invariant holds for $\S$, it must also hold for $\overline{\S}$.
\end{disjunction}
\end{proof}

\begin{corollary}\label{cor:invariant I'}
During the main loop of \TS, whenever a run is about to be pushed onto the stack, we have
$r_i \leqslant 2^{(i+1-j)/2} r_j$ for all integers $i \leqslant j \leqslant h$.
\end{corollary}

\begin{proof}
Since a run is about to be pushed, none of the conditions of Cases \#2 to \#5 hold in the stack $\S$.
Hence, we have $r_1 < r_2$, $r_1 + r_2 < r_3$ and $r_2 + r_3 < r_4$, and Lemma~\ref{lm:invariant I}
further proves that $r_i + r_{i+1} < r_{i+2}$ for all $i \in \{3,\ldots,h-2\}$.
In particular, for all $i \leqslant h-2$, we have $r_i < r_{i+1}$, and thus
$2 r_i \leqslant r_i + r_{i+1} \leqslant r_{i+2}$.
It follows immediately that $r_i \leqslant 2^{-k} r_{i+2k} \leqslant 2^{-k} r_{i+2k+1}$
for all integers $k \geqslant 0$, which is exactly the statement of Corollary~\ref{cor:invariant I'}.
\end{proof}

Corollary~\ref{cor:invariant I'} will be crucial in proving that the main loop of \TS can be performed
for a merge cost $\O(n + n \H)$. However, we do not prove this upper bound directly.
Instead, we need to distinguish several situations that may occur within the main loop.

Consider the sequence of Cases \#1 to \#5 triggered during the execution of the main loop of \TS.
It can be seen as a word on the alphabet $\{\#1,\ldots,\#5\}$ that starts with \#1, which completely encodes the execution of the algorithm. 
We split this word at every \#1, so that each piece corresponds to an iteration of the main loop. 
Those pieces are in turn split into two parts, at the first occurrence of a symbol \#3, \#4 or \#5.
The first half is called a \emph{starting sequence} and is made of a \#1 followed by the maximal number of \#2's. 
The second half is called an \emph{ending sequence}, it starts with \#3, \#4 or \#5 (or is empty) and it contains no occurrence of \#1 (see Figure~\ref{fig:sequence} for an example).

\begin{figure}[H]
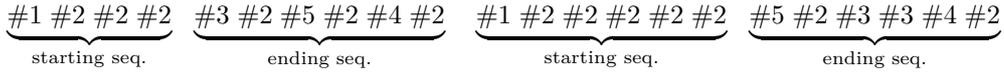

\smallskip
\centerline{$
\underbrace{\#1\;\#2\;\#2\;\#2}_{\text{starting seq.}}
~~\underbrace{\#3\;\#2\;\#5\;\#2\;\#4\;\#2}_{\text{ending seq.}}
~~~\underbrace{\#1\;\#2\;\#2\;\#2\;\#2\;\#2}_{\text{starting seq.}}
~~\underbrace{\#5\;\#2\;\#3\;\#3\;\#4\;\#2}_{\text{ending seq.}}
$}
\caption{The decomposition of the encoding of an execution into starting and ending sequences.\label{fig:sequence}}
\end{figure}

We bound the merge cost of starting sequences first, and will deal with ending sequences afterwards.

\begin{lemma}\label{lm:starting}
The cost of all merges performed during the starting sequences is~$\O(n)$.
\end{lemma}

\begin{proof}
More precisely, for a stack $\S=(R_1,\ldots,R_h)$, we prove that a starting sequence beginning with a push of a run~$R$ of size $r$ onto $\S$ uses at most $\gamma r$ comparisons in total,
where~$\gamma$ is the real constant $2 \sum_{j \geqslant 1} j / 2^{j/2}$. 
After the push, the stack is $\overline{\S} = (R,R_1,\ldots,R_h)$ and,
if the starting sequence contains $k \geqslant 1$ letters, i.e. $k-1$ occurrences of \#2, 
then this sequence amounts to merging the runs $R_1$, $R_2$, \ldots, $R_k$.
Since no merge is performed if $k = 1$, we assume below that $k \geqslant 2$.

More precisely, the total cost of these merges is
\[
C = (k-1)r_1+(k-1)r_2 + (k-2)r_3+\ldots + r_k \leqslant {\textstyle\sum_{i=1}^k} (k+1-i)r_i.
\]
The last occurrence of Case \#2 ensures that $r > r_k$,
hence applying Corollary~\ref{cor:invariant I'} to the stack $\S = (R_1,\ldots,R_h)$ shows that
$r \geqslant r_k \geqslant 2^{(k-1-i)/2} r_i$ for all $i = 1,\ldots,k$.
It follows that
\[
C / r \leqslant {\textstyle\sum_{i=1}^k} (k+1-i)2^{(i+1-k)/2} =
2 {\textstyle\sum_{j=1}^k} j 2^{-j/2} < \gamma.
\]

This concludes the proof, since each run is the beginning of exactly one starting sequence, and the sum of their lengths is $n$.
\end{proof}

Now, we must take care of run merges that take place during ending sequences.
The cost of merging two runs will be taken care of by making run elements pay tokens:
whenever two runs of lengths $r$ and $r'$ are merged, $r+r'$ tokens are paid
(not necessarily by the elements of those runs that are merged).
In order to do so, and to simplify the presentation, we also distinguish two kinds of tokens,
the \ctok-tokens and the \stok-tokens, which can both be used to pay for comparisons.

Two \ctok-tokens and one \stok-token are credited to an element when its run is pushed onto the stack
or when its height later decreases \emph{because of a merge that took place during an ending sequence}:
in the latter case, all the elements of $R_1$ are credited when $R_1$ and $R_2$ are merged, and
all the elements of $R_1$ and $R_2$ are credited when $R_2$ and $R_3$ are merged.
Tokens are spent to pay for comparisons, depending on 
the case triggered:
\begin{itemize}
\item Case \#2: every element of $R_1$ and $R_2$ pays 1 \ctok.
This is enough to cover the cost of merging $R_2$ and $R_3$, 
because $r_1 > r_3$ in this case, and therefore $r_2+r_1 \geqslant r_2 + r_3$.
\item Case \#3: every element of $R_1$ pays 2 \ctok. In this case $r_1\geqslant r_2$, and the cost is $r_1+r_2 \leqslant 2r_1$.
\item Cases \#4 and \#5: every element of $R_1$ pays 1 \ctok
and every element of $R_2$ pays 1 \stok. The cost $r_1+r_2$ is exactly the number of tokens spent.
\end{itemize}

\begin{lemma}\label{lm:balance}
The balances of \ctok-tokens and \stok-tokens of each element remain non-negative throughout  the main loop of \TS.
\end{lemma}

\begin{proof}
In all four cases \#2 to \#5, because the height of the elements of $R_1$ and possibly the height of those of $R_2$ decrease,
the number of credited \ctok-tokens after the merge is at least the number of \ctok-tokens spent.
The \stok-tokens are spent in Cases \#4 and \#5 only: every element of $R_2$ pays one \stok-token, and then belongs to the topmost run $\overline{R}_1$ of the new stack $\overline{\S}=(\overline{R}_1,\ldots, \overline{R}_{h-1})$ obtained after  merging $R_1$ and $R_2$. Since 
$\overline{R}_{i} = R_{i+1}$ for $i\geqslant 2$, the condition of Case~\#4 implies that $\overline{r}_1\geqslant \overline{r}_2$ and
the condition of Case~\#5 implies that $\overline{r}_1+\overline{r}_2\geqslant \overline{r}_3$: in both cases, the next modification of the stack $\overline{\S}$ is another merge, which belongs to the same ending sequence.

This merge decreases the height of $\overline{R}_1$, and therefore decreases the height of the elements of $R_2$, who will regain one \stok-token without losing any, since the topmost run of the stack never
pays with \stok-tokens. This proves that, whenever an element pay one \stok-token, the next modification is another merge during which it regains its \stok-token. This concludes the proof by direct induction.
\end{proof}

Finally, consider some element belonging to a run $R$.
Let $\S$ be the stack just before pushing the run $R$, and let
$\overline{S} = (\overline{R}_1,\ldots,\overline{R}_h)$ be the
stack just after the starting sequence of the run $R$
(i.e., the starting sequence initiated when $R$ is pushed onto $\S$) is over.
Every element of $R$ will be given at most $2h$ \ctok-tokens and $h$ \stok-tokens during the main loop
of the algorithm.

\begin{lemma}\label{lm:h-is-small}
The height of the stack when the starting sequence of the run $R$ is over
satisfies the inequality $h \leqslant 4 + 2 \log_2(n/r)$.
\end{lemma}

\begin{proof}
Since none of the runs $\overline{R}_3,\ldots,\overline{R}_h$
has been merged during the starting sequence of $R$,
applying Corollary~\ref{cor:invariant I'} to the stack $\S$ proves that
$\overline{r}_3 \leqslant 2^{2-h/2} \overline{r}_h \leqslant 2^{2-h/2} n$.
The run $R$ has not yet been merged either, which means that $r = \overline{r}_1$.
Moreover, at the end of this starting sequence, the conditions of case \#2 do not hold anymore,
which means that $\overline{r}_1 \leqslant \overline{r}_3$.
It follows that $r = \overline{r}_1 \leqslant \overline{r}_3 \leqslant 2^{2-h/2} n$,
which entails the desired inequality.
\end{proof}

Collecting all the above results is enough to prove Theorem~\ref{thm:complexity n + n H}.
First, as mentioned in Remark~\ref{rem:main-loop},
computing the run decomposition can be done in linear time. Then, we proved
that the starting sequences of the main loop have a merge cost $\O(n)$,
and that the ending sequences have a merge cost
$\O(\sum_{i=1}^\rho (1+\log(n/r_i))r_i) = \O(n + n \H)$.
Finally, the additional
merges of line~11 may be taken care of by Remark~\ref{rem:main-loop}.
This concludes the proof of the theorem.

\section{Refined analysis and precise worst-case complexity}\label{sec:analysis2}

The analysis performed in Section~\ref{sec:analysis1} proves that \TS sorts arrays in time
$\O(n + n \H)$. Looking more closely at the constants hidden in the $\O$ notation,
we may in fact prove that the cost of merges performed during an execution of \TS
is never greater than $6 n \H + \O(n)$.
However, the lower bound provided by Proposition~\ref{proposition:optimality}
only proves that the cost of these merges must be at least $n \H + \O(n)$.
In addition, there exist sorting algorithms~\cite{munro2018nearly}
whose merge cost is exactly $n \H + \O(n)$.

Hence, \TS is optimal only up to a multiplicative constant.
We focus now on finding the least real constant $\kappa$ such that
the merge cost of \TS is at most $\kappa n \H + \O(n)$,
thereby proving a conjecture of~\cite{BuKno18}.

\begin{theorem}\label{thm:complexity 1.5 n + n H}
The merge cost of \TS on arrays in $\C$ is at most $\kappa n \H + \O(n)$,
where $\kappa = 3/2$. Furthermore, $\kappa = 3/2$ is the least real constant
with this property.
\end{theorem}

The rest of this Section is devoted to proving Theorem~\ref{thm:complexity 1.5 n + n H}.
The theorem can be divided into two statements: one that states that \TS is asymptotically
optimal up to a multiplicative constant of $\kappa = 3/2$, and one that states that $\kappa$ is optimal.
The latter statement was proved in~\cite{BuKno18}. Here, we borrow their proof for the sake of completeness.

\begin{proposition}\label{pro:kappa-optimal-BuKno18}
There exist arrays of length $n$ on which the merge cost of \TS is at least $3/2 n \log_2(n) + \O(n)$.
\end{proposition}

\begin{proof}
The dynamics of \TS when sorting an array involves only the lengths of the monotonic runs in which
the array is split, not the actual array values. Hence, we identify every array with the sequence
of its run lengths. Therefore, every sequence of run lengths $\langle r_1,\ldots,r_\rho \rangle$ such that
$r_1,\ldots,r_{\rho-1} \geqslant 2$, $r_\rho \geqslant 1$ and $r_1+\ldots+r_\rho = n$ represents
at least one possible array of length $n$.

We define inductively a sequence of run lengths $\R(n)$ as follows:
\[\R(n) = \begin{cases}\langle n \rangle & \text{if } 1 \leqslant n \leqslant 6, \\
\R(k) \cdot \R(k-2) \cdot \langle 2\rangle & \text{if } n = 2k \text{ for some } k \geqslant 4, \\
\R(k) \cdot \R(k-1) \cdot \langle 2\rangle & \text{if } n = 2k+1 \text{ for some } k \geqslant 3,
\end{cases}\]
where the concanetation of two sequences $s$ and $t$ is denoted by $s \cdot t$.

Then, let us apply the main loop of TimSort
on an array whose associated monotonic runs
have lengths $\mathbf{r} = \langle r_1,\ldots,r_\rho \rangle$,
starting with an empty stack.
We denote the associated merge cost by $c(\mathbf{r})$ and,
if $\overline{\S} = (\overline{R}_1,\ldots,
\overline{R}_{\overline{h}})$ is the stack obtained
after the main loop has been applied,
we denote by $s(\mathbf{r})$ the sequence
$\langle \overline{r}_1,\ldots,\overline{r}_{\overline{h}}\rangle$.

An immediate induction shows that, if
$r_1 \geqslant r_2+\ldots+r_\rho+1$, then
$c(\mathbf{r}) = c(\langle r_2,\ldots,r_\rho\rangle)$ and
$s(\mathbf{r}) = \langle r_1 \rangle \cdot s(\langle r_2,\ldots,r_\rho\rangle)$. Similarly, if $r_1 \geqslant r_2+\ldots+r_\rho+1$ and
$r_2 \geqslant r_3+\ldots+r_\rho+1$, then
$c(\mathbf{r}) = c(\langle r_3,\ldots,r_\rho\rangle)$ and
$s(\mathbf{r}) = \langle r_1,r_2 \rangle \cdot s(\langle r_3,\ldots,r_\rho\rangle)$.

Consequently, and by another induction on $n$, it holds that
$s(\R(n)) = \langle n \rangle$ and that
\[c(\R(n)) = \begin{cases} 0 & \text{if } 1 \leqslant n \leqslant 6, \\
c(\R(k)) + c(\R(k-2)) + 3k & \text{if } n = 2k \text{ for some } k \geqslant 4, \\
c(\R(k)) + c(\R(k-1)) + 3k+2 & \text{if } n = 2k+1 \text{ for some } k \geqslant 3.
\end{cases}\]

Let $u_x = c(\R(\lfloor x \rfloor))$
and $v_x = (u_{x-4} - 15/2)/x - 3 \log_2(x) / 2$.
An immediate induction shows that $c(\R(n)) \geqslant c(\R(n+1))$
for all integers $n \geqslant 0$, which means that $x \mapsto u_x$
is non-decreasing. Then, we have 
$u_n = u_{n/2} + u_{(n-3)/2} + \lceil 3n/2 \rceil$
for all integers $n \geqslant 6$, and therefore
$u_x \geqslant 2 u_{x/2-2} + 3(x-1)/2$ for all real numbers $x \geqslant 6$. Consequently, for $x \geqslant 11$, it holds that
\[x v_x = u_{x-4} - 3 x \log_2(x) / 2 - 15/2 \geqslant
2 u_{x/2-4} + 3(x-5)/2 - 3 x \log_2(x) / 2 - 15/2 = x v_{x/2}.\]
This proves that $v_x \geqslant v_{x/2}$, from which it follows
that $v_x \geqslant \inf\{v_t \,:\, 11/2 \leqslant t < 11\}$.
Since $v_t = -15/(2t) -3 \log_2(t)/2 \geqslant -15/11 - 3 \log_2(11)/2 \geqslant -7$ for all $t \in [11/2,11)$,
we conclude that
$v_x \geqslant -7$ for all $x \geqslant 11$, and thus that
\[c(\R(n)) = u_n \geqslant (n+4) v_{n+4} + 3 (n+4) \log_2(n+4) / 2 \geqslant 3 n \log_2(n) / 2 - 7 (n+4),\] thereby proving
Proposition~\ref{pro:kappa-optimal-BuKno18}.
\end{proof}

It remains to prove the first statement of Theorem~\ref{thm:complexity 1.5 n + n H}. Our initial step towards
this statement consists in refining Lemma~\ref{lm:invariant I}.
This is the essence of Lemmas~\ref{lm:invariant II} to~\ref{lm:invariant IV}.

\begin{lemma}\label{lm:invariant II}
At any step during the main loop of \TS, if $h \geqslant 4$,
we have~$r_2<r_4$ and~$r_3<r_4$.
\end{lemma}

\begin{proof}
We proceed by induction. The proof consists in verifying that, if the
invariant holds at some point, then it still holds when an update of the stack occurs in one of the five situations labeled \#1 to \#5 in the algorithm.
This can be done by a straightforward case analysis.
We denote by $\S=(R_1,\ldots, R_h)$ the stack just before
the update, and by
$\overline{\S}=(\overline{R}_1,\ldots, \overline{R}_{\overline{h}})$ 
the new state of the stack after the update:
\begin{disjunction}
  \item If Case \#1 just occurred, a new run $\overline{R}_1$ was pushed.
  This implies that the conditions of Cases \#2 and \#4 did not hold in $\S$, otherwise merges would have continued. In particular, we have $\overline{r}_2=r_1<r_3=\overline{r}_4$ and
  $\overline{r}_3=r_2<r_1+r_2<r_3=\overline{r}_4$.
  
  \item If one of the Cases \#2 to \#5 just occurred, it holds that $\overline{r}_2 \leqslant r_2+r_3$, that $\overline{r}_3=r_4$ and that $\overline{r}_4=r_5$. Since Lemma~\ref{lm:invariant I} proves that $r_3+r_4<r_5$, it follows that
  $\overline{r}_2 \leqslant r_2+r_3 < r_3+r_4 < r_5 = \overline{r}_4$
  and that $\overline{r}_3 = r_4<r_3+r_4<r_5 = \overline{r}_4$.
\end{disjunction}
\end{proof}

\begin{lemma}\label{lm:invariant III}
At any step during the main loop of \TS, and for all
$i \in \{3,\ldots,h\}$, it holds that
$r_2+\ldots+r_{i-1} < \phi \, r_i$.
\end{lemma}

\begin{proof}
Like for Lemmas~\ref{lm:invariant I} and~\ref{lm:invariant II},
we proceed by induction and verify that, if the
invariant holds at some point, then it still holds when an update of the stack occurs in one of the five situations labeled \#1 to \#5 in the algorithm.
Let us denote by $\S=(R_1,\ldots, R_h)$ the stack just before
the update, and by
$\overline{\S}=(\overline{R}_1,\ldots, \overline{R}_{\overline{h}})$ 
the new state of the stack after the update:
\begin{disjunction}
  \item If Case \#1 just occurred, then we proceed by induction on
  $i \geqslant 3$. First, for $i = 3$,  since the conditions for Cases \#3 and \#4 do not hold in $\S$, we know that $\overline{r}_2 = r_1 < r_2 = \overline{r}_3$ and that
  $\overline{r}_2 + \overline{r}_3 = r_1+r_2 < r_3 = \overline{r}_4$.
  Then, for $i \geqslant 5$, Lemma~\ref{lm:invariant I} states that $r_{i-2}+r_{i-1} < r_i$, and therefore
  \begin{enumerate}[(i)]
   \item if $\overline{r}_{i-1} \leqslant \phi^{-1} \, \overline{r}_i$, then
  $\overline{r}_2+\ldots+\overline{r}_{i-1} < (\phi + 1) \overline{r}_{i-1} = \phi^2 \overline{r}_{i-1} \leqslant \phi \overline{r}_i$, and
  \item if $\overline{r}_{i-1} \geqslant \phi^{-1} \, \overline{r}_i$, then
  $\overline{r}_{i-2} \leqslant (1-\phi^{-1}) \ \overline{r}_i = \phi^{-2} \, \overline{r}_i$, and thus
  $\overline{r}_2+\ldots+\overline{r}_{i-1} < (\phi+1) \, \overline{r}_{i-2} + \overline{r}_{i-1} \leqslant
  \phi \, \overline{r}_{i-2} + \overline{r}_i \leqslant (\phi^{-1} + 1) \overline{r}_i = \phi \, \overline{r}_i$.
  \end{enumerate}
  Hence, in that case, it holds that $\overline{r}_2+\ldots+\overline{r}_{i-1} < \phi \, \overline{r}_i$ for all $i \in \{3,\ldots,h\}$.
  
  \item If one of the Cases \#2 to \#5 just occurred, it holds that $\overline{r}_2 \leqslant r_2+r_3$ and that $\overline{r}_j = r_{j+1}$ for all $j \geqslant 3$. It follows that
  $\overline{r}_2+\ldots+\overline{r}_{i-1} \leqslant
  r_2+\ldots+r_i < \phi \, r_{i+1} = \overline{r}_i$.
\end{disjunction}
\end{proof}

\begin{remark}
We could also have derived directly Lemma~\ref{lm:invariant II} from
Lemma~\ref{lm:invariant III}, by noting that
$\phi^2 \, r_2 = (\phi+1) r_2 < \phi \, r_2 + \phi \, r_3 < \phi^2 \, r_4$.
\end{remark}

\begin{lemma}\label{lm:invariant IV}
After every merge that occurred during an ending sequence,
we have~$r_1 < \phi^2 r_2$.
\end{lemma}

\begin{proof}
Once again, we proceed by induction.
We denote by $\S=(R_1,\ldots, R_h)$ the stack just before an update occurs, and by
$\overline{\S}=(\overline{R}_1,\ldots, \overline{R}_{\overline{h}})$ 
the new state of the stack after after the update:
\begin{disjunction}
  \item If Case \#2 just occurred, then
  this update is not the first one within the ending sequence,
  hence~$\overline{r}_1 = r_1 < \phi^2 \, r_2 < \phi^2 (r_2 + r_3) = \phi^2 \, \overline{r}_2$.
  
  \item If one of the Cases \#2 to \#5 just occurred,
  then $r_1 \leqslant r_3$ and Lemma~\ref{lm:invariant III} proves that
  $r_2 < \phi \, r_3$, which proves that
  $\overline{r}_1 = r_1 + r_2 < (\phi+1) r_3 = \phi^2 \, \overline{r}_2$.
\end{disjunction}
\end{proof}

\begin{lemma}\label{lm:invariant V}
After every merge triggered by Case $\#2$,
we have~$r_2 < \phi^2 r_1$.
\end{lemma}

\begin{proof}
We denote by $\S=(R_1,\ldots, R_h)$ the stack just before an update triggered by Case \#2 occurs, and by
$\overline{\S}=(\overline{R}_1,\ldots, \overline{R}_{\overline{h}})$
the new state of the stack after after the update.
It must hold that $r_1 > r_3$ and
Lemma~\ref{lm:invariant III} proves that $r_2 < \phi \, r_3$. It follows that
$\overline{r}_2 = r_2 + r_3 < (\phi+1) r_3 = \phi^2 \, r_3 < \phi^2 \, r_1 = \phi^2 \, \overline{r}_1$.
\end{proof}

Our second step towards proving the first statement of
Theorem~\ref{thm:complexity 1.5 n + n H}
consists in identifying which sequences of merges
an ending sequence may be made of.
More precisely, in the proof of Lemma~\ref{lm:balance},
we proved that every merge triggered by a case
$\#4$ or $\#5$ must be followed
by another merge, i.e.,
it cannot be the final merge of an ending sequence.

We present now a variant of this result,
which involves distinguishing between merges triggered
by a case $\#2$ and those triggered by a case $\#3$, $\#4$ or $\#5$.
Hence, we denote by $\caseX$ every $\#3$, $\#4$ or $\#5$.

\begin{lemma}\label{lm:(X2)*X*}
No ending sequence contains two conscutive $\#2$'s,
nor does it contain a subsequence of the form \caseX \caseX $\#2$.
\end{lemma}

\begin{proof}
Every ending sequence starts with an update \caseX,
where \caseX is equal to \#3, \#4 or \#5. Hence,
it suffices to prove that no ending sequence
contains a subsequence $\mathbf{t}$ of the form 
\caseX\caseX\#2 or \caseX\#2\;\#2.

Indeed, for the sake of contradiction, assume that it does,
and let $\S = (R_1,\ldots,R_h)$ be the stack just before
$\mathbf{t}$ starts.
We distinguish two cases, depending on the value of $\mathbf{t}$:
\begin{disjunction}
\item If $\mathbf{t}$ is the sequence $\caseX\;\caseX\;\#2$,
it must hold that $r_1+r_2 < r_4$ and that
$r_1+r_2+r_3 \geqslant r_5$, as illustrated in Figure~\ref{fig:XX2}~(top).
Since Lemma~\ref{lm:invariant I} proves that $r_3 + r_4 < r_5$,
it follows that $r_1+r_2+r_3 \geqslant r_5 > r_3+r_4 > r_1+r_2+r_3$,
which is impossible.
 
\item If $\mathbf{t}$ is the sequence $\caseX\;\#2\;\#2$, it must hold that
$r_1 < r_3$ and that
$r_1+r_2 \geqslant r_5$, as illustrated in Figure~\ref{fig:XX2}~(bottom).
Since Lemmas~\ref{lm:invariant I} and~\ref{lm:invariant II} prove that $r_3+r_4 < r_5$ and that $r_2 < r_4$,
it comes that
$r_1+r_2 \geqslant r_5 > r_3+r_4 > r_1+r_2$,
which is also impossible.\vspace{-0.2mm}
\end{disjunction}
\end{proof}

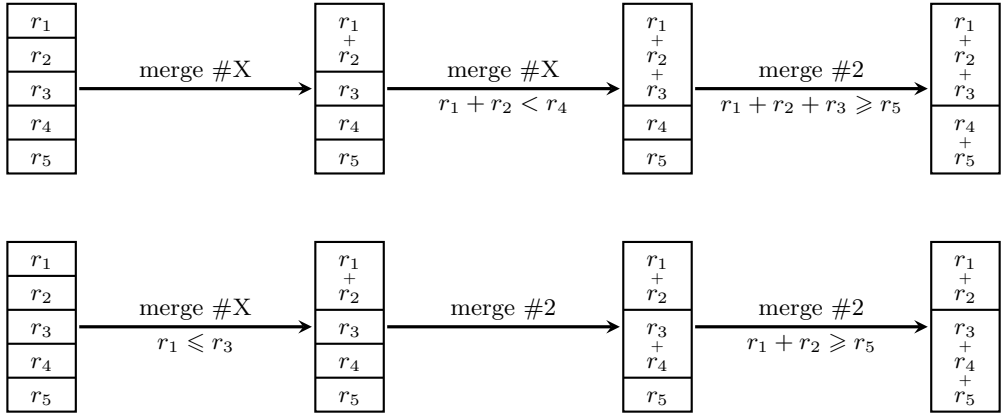
\begin{figure}[ht]
\begin{center}
\vspace{-2mm}
\begin{small}
\begin{tikzpicture}[scale=0.45]
\foreach \i in {0,...,3}{
 \draw[thick] (9*\i,0) -- (9*\i+2,0) -- (9*\i+2,5) -- (9*\i,5) -- cycle;
 \foreach \j in {1,...,5}{
   \node[anchor=south] at (9*\i+1,5-\j-0.05) {$r_\j$};
 }
}
\foreach \i in {0,...,2}{
 \FPeval{\k}{clip(4-\i)}
 \foreach \j in {1,...,\k}{
   \draw[thick] (9*\i,\j) -- (9*\i+2,\j);
 }
}
\draw[thick] (9*3,2) -- (9*3+2,2);

\draw[very thick,->,>=stealth] (2.1,2.5) -- (8.9,2.5);
\draw[very thick,->,>=stealth] (11.1,2.5) -- (17.9,2.5);
\draw[very thick,->,>=stealth] (20.1,2.5) -- (26.9,2.5);
\node[anchor=south] at (5.5,2.5-0.05) {merge \caseX};
\node[anchor=north] at (14.5,2.5+0.05) {$r_1+r_2<r_4$};
\node[anchor=south] at (14.5,2.5-0.05) {merge \caseX};
\node[anchor=north] at (23.5,2.5+0.05) {$r_1+r_2+r_3 \geqslant r_5$};
\node[anchor=south] at (23.5,2.5-0.05) {merge $\#2$};

\node[anchor=south] at (10,3.5) {\tiny+};
\node[anchor=south] at (19,3.5) {\tiny+};
\node[anchor=south] at (19,2.5) {\tiny+};
\node[anchor=south] at (28,3.5) {\tiny+};
\node[anchor=south] at (28,2.5) {\tiny+};
\node[anchor=south] at (28,0.5) {\tiny+};
\end{tikzpicture}

\bigskip\bigskip

\begin{tikzpicture}[scale=0.45]
\foreach \i in {0,...,3}{
 \draw[thick] (9*\i,0) -- (9*\i+2,0) -- (9*\i+2,5) -- (9*\i,5) -- cycle;
 \foreach \j in {1,...,5}{
   \node[anchor=south] at (9*\i+1,5-\j-0.05) {$r_\j$};
 }
}
\foreach \i in {0,...,1}{
 \FPeval{\k}{clip(4-\i)}
 \foreach \j in {1,...,\k}{
   \draw[thick] (9*\i,\j) -- (9*\i+2,\j);
 }
}
\draw[thick] (9*2,1) -- (9*2+2,1);
\draw[thick] (9*2,3) -- (9*2+2,3);
\draw[thick] (9*3,3) -- (9*3+2,3);

\draw[very thick,->,>=stealth] (2.1,2.5) -- (9*1-0.1,2.5);
\draw[very thick,->,>=stealth] (9*1+2.1,2.5) -- (9*2-0.1,2.5);
\draw[very thick,->,>=stealth] (9*2+2.1,2.5) -- (9*3-0.1,2.5);
\node[anchor=north] at (9*0.5+1,2.5+0.05) {$r_1 \leqslant r_3$};
\node[anchor=south] at (9*0.5+1,2.5-0.05) {merge \caseX};
\node[anchor=south] at (9*1+9*0.5+1,2.5-0.05) {merge $\#2$};
\node[anchor=north] at (9*2+9*0.5+1,2.5+0.05) {$r_1+r_2 \geqslant r_5$};
\node[anchor=south] at (9*2+9*0.5+1,2.5-0.05) {merge $\#2$};

\node[anchor=south] at (9*1+1,3.5) {\tiny+};
\node[anchor=south] at (9*2+1,3.5) {\tiny+};
\node[anchor=south] at (9*2+1,1.5) {\tiny+};
\node[anchor=south] at (9*3+1,3.5) {\tiny+};
\node[anchor=south] at (9*3+1,1.5) {\tiny+};
\node[anchor=south] at (9*3+1,0.5) {\tiny+};
\end{tikzpicture}
\end{small}
\end{center}
\vspace{-2mm}
\caption{Applying successively merges
\caseX\#2\;\#2
or \caseX\caseX\#2
to a stack is impossible.}
\label{fig:XX2}
\end{figure}

Our third step consists in modifying the cost allocation
we had chosen in Section~\ref{sec:analysis1}, which
is not sufficient to prove Theorem~\ref{thm:complexity 1.5 n + n H}.
Instead, we associate to
every run $R$ its \emph{potential}, which depends only on the length
$r$ of the run, and is defined as
$\pot(r) = 3 r \log_2(r) / 2$. We also call \emph{potential} of a
set of runs the sum of the potentials of the runs it is formed of,
and \emph{potential variation} of a (sequence of) merges
the increase in potential caused by these merge(s).

We shall prove that the potential variation of every ending
sequence dominates its merge cost, up to a small error term.
In order to do this, let us study more precisely individual merges.
Below, we respectively denote by $\Delta_{\pot}(\mathbf{m})$
and $\cost(\mathbf{m})$ the potential variation and the
merge cost of a merge $\mathbf{m}$. Then, we say that
$\mathbf{m}$ is a \emph{balanced} merge if
$\cost(\mathbf{m}) \leqslant \Delta_{\pot}(\mathbf{m})$.

In the next Lemmas, we prove that most merges are balanced
or can be grouped into sequences of merges that are balanced
overall.

\begin{lemma}\label{lm:delta-cost I}
Let $\mathbf{m}$ be a merge between two runs $R$ and $R'$.
If $\phi^{-2} \, r \leqslant r' \leqslant \phi^2 \, r$,
then $\mathbf{m}$ is balanced.
\end{lemma}

\begin{proof}
Let $x = r / (r+r')$: we have~$\Phi < x < 1-\Phi$, where $\Phi = 1/(1+\phi^2)$.
Then, observe that $\Delta(\mathbf{m}) = 3 (r+r') H(x) / 2$,
where $H(x) = - x \log_2(x) - (1-x) \log_2(x)$
is the binary Shannon entropy of a Bernoulli law of parameter $x$.
Moreover, the function $z \mapsto H(z) = H(1-z)$
is increasing on $[0,1/2]$. It follows that $H(x) \geqslant
H(\Phi) \approx 0.85 > 2/3$, and therefore that $\Delta(\mathbf{m}) > r+r' = \cost(\mathbf{m})$. 
\end{proof}

\begin{lemma}\label{lm:delta-cost II}
Let $\mathbf{m}$ be a merge that belongs to some ending sequence.
If $\mathbf{m}$ is a merge $\#2$, then $\mathbf{m}$ is balanced and,
if $\mathbf{m}$ is followed by another merge $\mathbf{m}'$,
then $\mathbf{m}'$ is also balanced.
\end{lemma}

\begin{proof}
Lemma~\ref{lm:(X2)*X*} ensures that
$\mathbf{m}$ was preceded by another merge $\mathbf{m}^\star$,
which must be a merge \caseX.
Denoting by $\S = (R_1,\ldots,R_h)$ the stack just before the merge
$\mathbf{m}^\star$ occurs,
the update $\mathbf{m}$ consists in merging
the runs $R_3$ and $R_4$.
Then, it comes that $r_1 \leqslant r_3$ and that $r_1+r_2 > r_4$,
while Lemma~\ref{lm:invariant II} and~\ref{lm:invariant III} respectively
prove that $r_3 < r_4$ and that $r_2 < \phi \, r_3$.
Hence, we both have $r_3 < r_4$ and
$r_4 < r_1 + r_2 < (1 + \phi) r_3 = \phi^2 \, r_3$, and Lemma~\ref{lm:delta-cost II} proves that $\mathbf{m}$ is balanced.

Then, if $\mathbf{m}$ is followed by another merge $\mathbf{m}'$,
Lemma~\ref{lm:(X2)*X*} proves that $\mathbf{m}'$ is also a merge \caseX, between runs of respective lengths $r_1+r_2$ and $r_3+r_4$.
Note that $r_1 \leqslant r_3$ and that $r_1+r_2 > r_4$. Since
Lemma~\ref{lm:invariant II} proves that $r_2 < r_4$ and that $r_3 < r_4$, it follows that
$2(r_1+r_2) > 2 r_4 > r_3 + r_4 > r_1+r_2$ and,
using the fact that $2 < 1 + \phi = \phi^2$,
Lemma~\ref{lm:delta-cost II} therefore proves that $\mathbf{m}$ is balanced.
\end{proof}

\begin{lemma}\label{lm:delta-cost III}
Let $\mathbf{m}$ be a merge \caseX
between two runs $R_1$ and $R_2$ such that $r_1 < \phi^{-2} \, r_2$.
Then, $\mathbf{m}$ is followed by another merge $\mathbf{m}'$, and
$\cost(\mathbf{m}) + \cost(\mathbf{m}') \leqslant \Delta_{\pot}(\mathbf{m}) + \Delta_{\pot}(\mathbf{m}')$.
\end{lemma}

\begin{proof}
Let $\mathbf{m}^\star$ be the update the immediately precedes
$\mathbf{m}$. Let also $\S^\star = (R^\star_1,\ldots,R^\star_{h^\star})$, $\S = (R_1,\ldots,R_h)$ and
$\S' = (R'_1,\ldots,R'_{h'})$ be the respective states
of the stack just before $\mathbf{m}^\star$ occurs,
just before $\mathbf{m}$ occurs and just after $\mathbf{m}$ occurs.

Since $r_1 < \phi^{-2} \, r_2$,
Lemma~\ref{lm:invariant V} proves that
$\mathbf{m}^\star$ is either an update \#1 or a merge \caseX.
In both cases, it follows that $r_2 < r_3$ and that $r_2+r_3 < r_4$.
Indeed, if $\mathbf{m}^\star$ is an update \#1, then we must have
$r_2 = r^\star_1 < r^\star_2 = r_3$ and
$r_2+r_3 = r^\star_1+r^\star_2 < r^\star_3 = r_4$, and if
$\mathbf{m}'$ is a merge \caseX, then Lemmas~\ref{lm:invariant I} and~\ref{lm:invariant II}
respectively prove that $r_2+r_3 = r^\star_3+r^\star_4 < r^\star_5 = r_4$ and that $r_2 = r^\star_3 < r^\star_4 = r_3$.

Then, since $\mathbf{m}$ is a merge \caseX, we also know that
$r_1 \leqslant r_3$. Since $r_1 < \phi^{-2} \, r_2$ and $r_2+r_3 < r_4$, this means that $r_1 + r_2 \geqslant r_3$.
It follows that
$r'_2 = r_3 \leqslant r_1+ r_2 = r'_1$ and that
$r'_1 = r_1 + r_2 \leqslant r_2 + r_3 < r_4 = r'_3$.
Consequently, the merge $\mathbf{m}$ must be followed by a merge
$\mathbf{m}'$, which is triggered by case \#3.

Finally, let $x = r_1 / (r_1+r_2)$ and $y = (r_1+r_2) / (r_1+r_2+r_3)$.
It comes that $\cost(\mathbf{m})+\cost(\mathbf{m}') = (r_1+r_2+r_3)(1+y)$ and that
$\Delta_{\pot}(\mathbf{m}) + \Delta_{\pot}(\mathbf{m}') = 3 (r_1+r_2+r_3) \left(y H(x) + H(y)\right)\!/ 2$, where
we recall that $H$ is the binary Shannon entropy function, with
$H(t) = - t \log_2(t) - (1-t) \log_2(t)$.
The above inequalities about $r_1$, $r_2$ and $r_3$ prove that $0 \leqslant 2 - 1/y \leqslant x \leqslant 1/(1+\phi^2)$.
Since $H$ is increasing on the interval $[0,1/2]$, and since
$1+\phi^2 \geqslant 2$, it follows that
$\Delta_{\pot}(\mathbf{m}) + \Delta_{\pot}(\mathbf{m}') \geqslant 3 (r_1+r_2+r_3) \left(y H(2 - 1/y) + H(y)\right)\!/ 2$.

Hence, let $F(y) = 3 \left(y H(2 - 1/y) + H(y)\right)\!/2 - (1+y)$.
We shall prove that $F(y) \geqslant 0$ for all $y \geqslant 0$
such that
$0 \leqslant 2 - 1/y \leqslant 1/(1+\phi^2)$, i.e., such that
$1/2 \leqslant y \leqslant 
(1+\phi^2)/(1+2\phi^2)$.
To that mean, observe that
$F''(y) = 3 /\!\left((1-y)(1-2y) \ln(2)\right) < 0$
for all $y \in (1/2,1)$. Thus, $F$ is concave on
$(1/2,1)$.
Since $F(1/2) = 0$ and $F(3/4) = 1/2$, it follows that
$F(y) \geqslant 0$ for all $y \in [1/2,3/4]$.
Checking that $(1+\phi^2)/(1+2\phi^2) < 3/4$ completes the proof.
\end{proof}

\begin{lemma}\label{lm:delta-cost IV}
Let $\mathbf{m}$ be the first merge of the ending sequence
associated with a run $R$. Let $R_1$ and $R_2$ be the runs
that $\mathbf{m}$ merges together.
If $r_1 > \phi^2 \, r_2$, it holds that
$\cost(\mathbf{m}) \leqslant \Delta_{\pot}(\mathbf{m}) + r$.
\end{lemma}

\begin{proof}
By definition of $\mathbf{m}$,
we have~$R = R_1$, and thus $r = r_1 \geqslant r_2$.
Hence, it follows that
$\Delta_{\pot}(\mathbf{m}) = r \log((r+r_2)/r) +
r_2 \log((r+r_2)/r_2) \geqslant r_2 \log((r+r_2)/r_2) \geqslant
r_2 = \cost(\mathbf{m}) - r$.
\end{proof}

\begin{proposition}\label{pro:delta-cost}
Let $\mathbf{s}$ be the ending sequence associated with a run $R$,
and let $\Delta_{\pot}(\mathbf{s})$ and $\cost(\mathbf{s})$
be its potential variation and its merge cost.
It holds that~$\cost(\mathbf{s}) \leqslant \Delta_{\pot}(\mathbf{s}) + r$.
\end{proposition}

\begin{proof}
Let us group the merges of $\mathbf{s}$ as follows:
\begin{enumerate}[(i)]
 \item if $\mathbf{m}$ is an unbalanced merge \caseX between two runs $R_1$ and $R_2$ such that $r_1 < r_2$, then $\mathbf{m}$ is followed by another merge $\mathbf{m}'$, and we group $\mathbf{m}$
 and $\mathbf{m}'$ together;\label{delta-cost:case:1}
 \item otherwise, and if $\mathbf{m}$ has not been grouped
 with its predecessor, it forms its own group.
\end{enumerate}
In case~\eqref{delta-cost:case:1}, Lemma~\ref{lm:delta-cost IV} ensures that $\mathbf{m}'$ itself cannot be grouped with another merge.
This means that our grouping is unambiguous.

Then, let $\mathbf{g}$ be such a group,
with potential variation $\Delta_{\pot}(\mathbf{g})$ and
merge cost $\cost(\mathbf{g})$.
Lemmas~\ref{lm:delta-cost I} to~\ref{lm:delta-cost IV}
prove that $\cost(\mathbf{g}) \leqslant \Delta_{\pot}(\mathbf{g}) + r$ if $\mathbf{g}$ is formed of the first merge of $\mathbf{s}$ only,
and that $\cost(\mathbf{g}) \leqslant \Delta_{\pot}(\mathbf{g})$
in all other cases. Proposition~\ref{pro:delta-cost} follows.
\end{proof}

Collecting all the above results is enough to prove Theorem~\ref{thm:complexity 1.5 n + n H}.
First, like in Section~\ref{sec:analysis1},
computing the run decomposition and merging runs in
starting sequences has a cost $\O(n)$,
and the final merges of line~11 may be taken care of by Remark~\ref{rem:main-loop}.
Second, by Proposition~\ref{pro:delta-cost},
ending sequences have a merge cost dominated by
$\Delta_{\pot} + n$, where $\Delta_{\pot}$
is the total variation of potential during the algorithm.
Observing that
$\Delta_{\pot} = -3/2 \sum_{i=1}^\rho r_i \log_2(r_i/n) = - 3 n \H / 2$
concludes the proof of the theorem.

\section{About the Java version of TimSort}\label{sec:java}

\begin{figure}[ht]
\begin{center}
\begin{small}
\setlength{\tabcolsep}{1pt}
\begin{tabular}{cccccccccccc}
\begin{tikzpicture}[every node/.style={align=center,text width=1em,inner sep=2.5pt}]
\matrix (A) [matrix of nodes,nodes={draw}] { {\bf 24}\\ };
\node [above, yshift=-1mm] at (A.north) {$\#1$};
\end{tikzpicture}
&
\begin{tikzpicture}[every node/.style={align=center,text width=1em,inner sep=2.5pt}]
\matrix (A) [matrix of nodes,nodes={draw}] { {\bf 18}\\24\\ };
\node [above, yshift=-1mm] at (A.north) {$\#1$};
\end{tikzpicture}
&
\begin{tikzpicture}[every node/.style={align=center,text width=1em,inner sep=2.5pt}]
\matrix (A) [matrix of nodes,nodes={draw}] { {\bf 50}\\18\\24\\ };
\node [above, yshift=-1mm] at (A.north) {$\#1$};
\end{tikzpicture}
&
\begin{tikzpicture}[every node/.style={align=center,text width=1em,inner sep=2.5pt}]
\matrix (A) [matrix of nodes,nodes={draw}] { 50\\42\\ };
\node [above, yshift=-1mm] at (A.north) {$\#2$};
\end{tikzpicture}
&
\begin{tikzpicture}[every node/.style={align=center,text width=1em,inner sep=2.5pt}]
\matrix (A) [matrix of nodes,nodes={draw}] { 92\\ };
\node [above, yshift=-1mm] at (A.north) {$\#3$};
\end{tikzpicture}
&
\begin{tikzpicture}[every node/.style={align=center,text width=1em,inner sep=2.5pt}]
\matrix (A) [matrix of nodes,nodes={draw}] { {\bf 28}\\92\\ };
\node [above, yshift=-1mm] at (A.north) {$\#1$};
\end{tikzpicture}
&
\begin{tikzpicture}[every node/.style={align=center,text width=1em,inner sep=2.5pt}]
\matrix (A) [matrix of nodes,nodes={draw}] { {\bf 20}\\28\\92\\ };
\node [above, yshift=-1mm] at (A.north) {$\#1$};
\end{tikzpicture}
&
\begin{tikzpicture}[every node/.style={align=center,text width=1em,inner sep=2.5pt}]
\matrix (A) [matrix of nodes,nodes={draw}] { {\bf 6}\\20\\28\\92\\ };
\node [above, yshift=-1mm] at (A.north) {$\#1$};
\end{tikzpicture}
&
\begin{tikzpicture}[every node/.style={align=center,text width=1em,inner sep=2.5pt}]
\matrix (A) [matrix of nodes,nodes={draw}] { {\bf 4}\\6\\20\\28\\92\\ };
\node [above, yshift=-1mm] at (A.north) {$\#1$};
\end{tikzpicture}
&
\begin{tikzpicture}[every node/.style={align=center,text width=1em,inner sep=2.5pt}]
\matrix (A) [matrix of nodes,nodes={draw}] { {\bf 8}\\4\\6\\20\\28\\92\\ };
\node [above, yshift=-1mm] at (A.north) {$\#1$};
\end{tikzpicture}
&
\begin{tikzpicture}[every node/.style={align=center,text width=1em,inner sep=2.5pt},row 2 column 1/.style={text=red},row 3 column 1/.style={text=red},row 4 column 1/.style={text=red}]
\matrix (A) [matrix of nodes,nodes={draw}] { 8\\10\\20\\28\\92\\ };
\node [above, yshift=-1mm] at (A.north) {$\#2$};
\end{tikzpicture}
&
\begin{tikzpicture}[every node/.style={align=center,text width=1em,inner sep=2.5pt}]
\matrix (A) [matrix of nodes,nodes={draw}] { {\bf 1}\\8\\10\\20\\28\\92\\ };
\node [above, yshift=-1mm] at (A.north) {$\#1$};
\end{tikzpicture}
\end{tabular}
\end{small}
\end{center}
\vspace{-6mm}
\caption{Execution of the main loop of Java's \TS (Algorithm~\ref{alg:TS translated}, without merge case \#5, at line~\ref*{algline:new_cond_translated}), 
with the lengths of the runs in $\rundecomp$ being $(24, 18, 50, 28, 20, 6, 4, 8, 1)$. 
When the second to last run (of length~8) is pushed onto the stack, the while loop of line~\ref*{algline:inner_while_translated} stops after 
only one merge, breaking the invariant (in red), unlike what we see in~Figure~\ref{fig:ts-python-exec} using the Python version of \TS.}
\label{fig:invariant_bug}
\end{figure}
  
Algorithm~\ref{alg:merge_collapse} (and therefore Algorithm~\ref{alg:TS translated}) does not correspond to the original \TS. 
Before release 3.4.4 of Python, the second part of the condition (in blue) in the test at line~\ref*{algline:new_cond} of {\tt merge\_collapse} (and therefore merge case \#5 of Algorithm~\ref{alg:TS translated}) was missing. This version 
of the algorithm worked fine, meaning that it did actually sort arrays, but the 
invariant given by Equation~\eqref{eq:inv1} did not hold.
Figure~\ref{fig:invariant_bug} illustrates the difference caused by the missing condition when running Algorithm~\ref{alg:TS translated} on the same input as in~Figure~\ref{fig:ts-python-exec}. 

This was discovered by de Gouw {\em et al.}~\cite{GoRoBoBuHa15} when trying to prove the correctness of the Java implementation of \TS (which is the same as in the earlier versions of Python). And since the Java version of the algorithm uses the (wrong) invariant to compute the maximum size of the stack used to store the runs, the authors were able to build a sequence of runs that causes the Java implementation of \TS to crash.
They proposed two solutions to fix \TS: reestablish the invariant, which led to the current Python version,
or keep the original algorithm and compute correct bounds for the stack size, which is the solution that was chosen in Java~9 (note that this is the second time these values had to be changed). 
To do the latter, the developers used the claim in~\cite{GoRoBoBuHa15} that the invariant cannot be violated for two consecutive runs on the stack, which turns out to be false,\footnote{This is the consequence of a small error in the proof of their Lemma~1. The constraint $C_1>C_2$ has no reason to be. Indeed, in our example, we have $C_1=25$ and $C_2=31$.} as illustrated in Figure~\ref{fig:inv_still_broken}.
Thus, it is still possible to cause the Java implementation to fail:
it uses a stack of runs of size at most~49
and we were able to compute an example
requiring a stack of size~50 (see~\url{http://igm.univ-mlv.fr/~pivoteau/Timsort/Test.java}), causing an error at runtime in Java's sorting method.

\begin{figure}[h]
\begin{center}
\begin{small}
\setlength{\tabcolsep}{1pt}
\begin{tabular}{ccccccccccccc}
\begin{tikzpicture}[every node/.style={align=center,text width=15pt,inner sep=2.5pt}]
\matrix (A) [matrix of nodes,nodes={draw}] { {\bf 109}\\ };
\node [above, yshift=-1mm] at (A.north) {$\#1$};
\end{tikzpicture}
&
\begin{tikzpicture}[every node/.style={align=center,text width=15pt,inner sep=2.5pt}]
\matrix (A) [matrix of nodes,nodes={draw}] { {\bf 83}\\109\\ };
\node [above, yshift=-1mm] at (A.north) {$\#1$};
\end{tikzpicture}
&
\begin{tikzpicture}[every node/.style={align=center,text width=15pt,inner sep=2.5pt}]
\matrix (A) [matrix of nodes,nodes={draw}] { {\bf 25}\\83\\109\\ };
\node [above, yshift=-1mm] at (A.north) {$\#1$};
\end{tikzpicture}
&
\begin{tikzpicture}[every node/.style={align=center,text width=15pt,inner sep=2.5pt}]
\matrix (A) [matrix of nodes,nodes={draw}] { {\bf 16}\\25\\83\\109\\ };
\node [above, yshift=-1mm] at (A.north) {$\#1$};
\end{tikzpicture}
&
\begin{tikzpicture}[every node/.style={align=center,text width=15pt,inner sep=2.5pt}]
\matrix (A) [matrix of nodes,nodes={draw}] { {\bf 8}\\16\\25\\83\\109\\ };
\node [above, yshift=-1mm] at (A.north) {$\#1$};
\end{tikzpicture}
&
\begin{tikzpicture}[every node/.style={align=center,text width=15pt,inner sep=2.5pt}]
\matrix (A) [matrix of nodes,nodes={draw}] { {\bf 7}\\8\\16\\25\\83\\109\\ };
\node [above, yshift=-1mm] at (A.north) {$\#1$};
\end{tikzpicture}
&
\begin{tikzpicture}[every node/.style={align=center,text width=15pt,inner sep=2.5pt}]
\matrix (A) [matrix of nodes,nodes={draw}] { {\bf 26}\\7\\8\\16\\25\\83\\109\\ };
\node [above, yshift=-1mm] at (A.north) {$\#1$};
\end{tikzpicture}
&
\begin{tikzpicture}[every node/.style={align=center,text width=15pt,inner sep=2.5pt}]
\matrix (A) [matrix of nodes,nodes={draw}] { 26\\15\\16\\25\\83\\109\\ };
\node [above, yshift=-1mm] at (A.north) {$\#2$};
\end{tikzpicture}
&
\begin{tikzpicture}[every node/.style={align=center,text width=15pt,inner sep=2.5pt}]
\matrix (A) [matrix of nodes,nodes={draw}] { 26\\31\\25\\83\\109\\ };
\node [above, yshift=-1mm] at (A.north) {$\#2$};
\end{tikzpicture}
&
\begin{tikzpicture}[every node/.style={align=center,text width=15pt,inner sep=2.5pt}]
\matrix (A) [matrix of nodes,nodes={draw}] { 26\\56\\83\\109\\ };
\node [above, yshift=-1mm] at (A.north) {$\#2$};
\end{tikzpicture}
&
\begin{tikzpicture}[every node/.style={align=center,text width=15pt,inner sep=2.5pt}]
\matrix (A) [matrix of nodes,nodes={draw}] { {\bf 2}\\26\\56\\83\\109\\ };
\node [above, yshift=-1mm] at (A.north) {$\#1$};
\end{tikzpicture}
&
\begin{tikzpicture}[every node/.style={align=center,text width=15pt,inner sep=2.5pt}]
\matrix (A) [matrix of nodes,nodes={draw}] { {\bf 27}\\2\\26\\56\\83\\109\\ };
\node [above, yshift=-1mm] at (A.north) {$\#1$};
\end{tikzpicture}
&
\begin{tikzpicture}[every node/.style={align=center,text width=15pt,inner sep=2.5pt},row 2 column 1/.style={text=red},row 3 column 1/.style={text=red},row 4 column 1/.style={text=red},row 5 column 1/.style={text=red}]
\matrix (A) [matrix of nodes,nodes={draw}] { 27\\28\\56\\83\\109\\ };
\node [above, yshift=-1mm] at (A.north) {$\#2$};
\end{tikzpicture}
\end{tabular}
\end{small}
\end{center}
\vspace{-6mm}
\caption{
Execution of the main loop of the Java version of \TS (without merge case \#5, at line~\ref*{algline:new_cond_translated} of Algorithm~\ref{alg:TS translated}), with the lengths of the runs in $\rundecomp$ 
being $(109, 83, 25, 16, 8, 7, 26, 2, 27)$. When the algorithm stops, the invariant is violated twice, for consecutive runs (in red).}
\label{fig:inv_still_broken}
\end{figure}

Even if the bug we highlighted in Java's \TS is very unlikely to happen, this should be corrected.
And, as advocated by de Gouw {\em et al.} and Tim Peters himself,\footnote{Here is the discussion about the correction in Python: \url{https://bugs.python.org/issue23515}.} we strongly believe that the best solution would be to correct the algorithm as in the current version of Python, in order to keep it clean and simple.
However, since this is the implementation of Java's sort for the moment, there are two questions we would like to tackle: Does the complexity analysis holds without the missing condition? And, can we compute an actual bound for the stack size?
We first address the complexity question. It turns out that the missing invariant was a key ingredient for having a simple and elegant proof.

\begin{proposition}\label{pro:domination}
\reversemarginpar
\marginnote{\rm Full proof in Section~\ref{sec:domination}.}
At any time during the main loop of Java's \TS, if the stack of runs is
$(R_1,\ldots, R_h)$ then we have
$r_3<r_4<\ldots<r_h$ and, for all $i \geqslant 3$, we have $(2 + \sqrt{7}) r_i \geqslant r_2 + \ldots + r_{i-1}$.
\end{proposition}

\begin{proof}[Proof ideas]
The proof of Proposition~\ref{pro:domination} is much more technical and difficult than insightful, and therefore
we just summarize its main steps.
As in previous sections, this proof relies on several inductive arguments, using both inductions
on the number of merges performed, on the stack size and on the run sizes.
The inequalities $r_3<r_4<\ldots<r_h$ come at once, hence we focus on the second part of Proposition~\ref{pro:domination}.

Since separating starting and ending sequences was useful in Section~\ref{sec:analysis2},
we first introduce the notion of \emph{stable} stacks: a stack $\S$ is stable if,
when operating on the stack $\S = (R_1,\ldots,R_h)$,
Case \#1 is triggered (i.e. Java's \TS is about to perform a \emph{run push} operation).

We also call \emph{obstruction indices} the integers $i \geqslant 3$
such that $r_i \leqslant r_{i-1} + r_{i-2}$: although they do not exist in Python's \TS,
they may exist, and even be consecutive, in Java's \TS.
We prove that, if $i-k,i-k+1,\ldots,i$ are obstruction indices, then the stack sizes $r_{i-k-2},\ldots,r_i$
grow ``at linear speed''. For instance, in the last stack of Figure~\ref{fig:inv_still_broken}, obstruction indices are $4$ and $5$,
and we have $r_2 = 28$, $r_3 = r_2 + 28$, $r_4 = r_3 + 27$ and $r_5 = r_4 + 26$.

Finally, we study so-called \emph{expansion functions}, i.e.
functions $f : [0,1] \mapsto \mathbb{R}$ such that, for every
stable stack $\S = (R_1,\ldots, R_h)$, we have $r_2+\ldots+r_{h-1} \leqslant r_h f(r_{h-1} / r_h)$.
We exhibit an explicit function $f$ such that $f(x) \leqslant 2+\sqrt{7}$ for all $x \in [0,1]$,
and we prove by induction on $r_h$ that $f$ is an expansion function,
from which we deduce Proposition~\ref{pro:domination}.
\end{proof}

Once Proposition~\ref{pro:domination} is proved, we easily recover the following variant
of Lemmas~\ref{lm:invariant I} and~\ref{lm:h-is-small}.

\begin{lemma}\label{lem:height-J-Java}
At any time during the main loop of Java's \TS, if the stack is
$(R_1,\ldots, R_h)$ then we have 
$r_2 / (2 + \sqrt{7}) \leqslant r_3 < r_4 < \ldots < r_h$ and, for all $i \geqslant j \geqslant 3$, we have $r_i \geqslant \delta^{i-j-4}r_j$,
where $\delta = \left( 5/(2+\sqrt{7}) \right)^{1/5} > 1$. 
Furthermore, at any time during an ending sequence, including just before it starts and just after it ends,
we have $r_1 \leqslant (2 + \sqrt{7}) r_3$. 
\end{lemma}

\begin{proof}
The inequalities $r_2 / (2 + \sqrt{7}) \leqslant r_3 < r_4 < \ldots < r_h$ are just a (weaker) restatement of Proposition~\ref{pro:domination}.
Then, for $j \geqslant 3$, we have $(2+\sqrt{7}) r_{j+5} \geqslant r_j+\ldots+r_{j+4} \geqslant 5 r_j$, i.e. $r_{j+5} \geqslant \delta^5 r_j$,
from which one gets that $r_i \geqslant \delta^{i-j-4}r_j$.

Finally, we prove by induction
that $r_1 \leqslant (2 + \sqrt{7}) r_3$ during ending sequences.
First, when the ending sequence starts, $r_1 < r_3 \leqslant (2+\sqrt{7}) r_3$.
Before any merge during this sequence, if the stack is $\S=(R_1,\ldots R_h)$, then we denote by $\overline{\S} = (\overline{R}_1,\ldots,\overline{R}_{h-1})$
the stack after the merge.
If the invariant holds before the merge, in Case \#2, 
we have $\overline{r}_1 = r_1 \leqslant (2+\sqrt{7}) r_3 \leqslant (2+\sqrt{7}) r_4 = (2+\sqrt{7}) \overline{r}_3$; 
and using Proposition~\ref{pro:domination} in Cases \#3 and~\#4, 
we have $\overline{r}_1 = r_1 + r_2$ and $r_1 \leqslant r_3$, 
hence $\overline{r}_1 = r_1 + r_2 \leqslant r_2 + r_3 \leqslant (2+\sqrt{7}) r_4 = (2+\sqrt{7}) \overline{r}_3$, concluding the proof.
\end{proof}

We can then recover a proof of complexity for the Java version of \TS,
by following the same proof as in Sections~\ref{sec:analysis1} and~\ref{sec:analysis2},
but using Lemma~\ref{lem:height-J-Java} instead of Lemmas~\ref{lm:invariant I} and~\ref{lm:h-is-small}.

\begin{theorem}\label{thm:complexity n log rho - java}
The complexity of Java's \TS on inputs of size $n$ with $\rho$ runs is $\O(n + n \log \rho )$.
\end{theorem}

Another question is that of the stack size requirements of Java's \TS,
i.e. computing~$h_{\max}$.
A first result is the following immediate corollary of Lemma~\ref{lem:height-J-Java}.

\begin{corollary}\label{cor:bad-size}
  On an input of size $n$, Java's \TS will create a stack of runs of maximal size $h_{\max} \leqslant 7 + \log_\delta(n)$, 
  where $\delta = \left( 5/(2+\sqrt{7}) \right)^{1/5}$.
\end{corollary}

\begin{proof}
At any time during the main loop of Java's \TS on an input of size $n$, if the stack is
$(R_1,\ldots, R_h)$ and $h \geqslant 3$, it follows from Lemma~\ref{lem:height-J-Java} that
$n \geqslant r_h \geqslant \delta^{h-7} r_3 \geqslant \delta^{h-7}$.
\end{proof}

Unfortunately, for integers smaller than $2^{31}$, Corollary~\ref{cor:bad-size} only proves that the stack size will never exceed~$347$. 
However, in the comments of Java's implementation of \TS,\footnote{Comment at line~168: \url{http://igm.univ-mlv.fr/~pivoteau/Timsort/TimSort.java}.} there is a remark that keeping a short stack is of some importance, for practical reasons, and that the value chosen in Python --~$85$~-- is ``too expensive''. 
Thus, in the following, we go to the extent of computing the optimal bound. 
It turns out that this bound cannot exceed~$86$ for such integers.
This bound could possibly be refined slightly, but definitely not to the point of competing
with the bound that would be obtained if the invariant of Equation~\eqref{eq:inv1} were correct.
Once more, this suggests that implementing the new version of \TS in Java would be a good idea,
as the maximum stack height is smaller in this case.

\begin{theorem}\label{thm:good-size}
\marginnote{\rm Full proof in\phantom{xx} Sections~\ref{sec:good1}\phantom{xx}  and~\ref{sec:good2}.\phantom{xx} }
On an input of size $n$, Java's \TS will create a stack of runs of maximal size $h_{\max} \leqslant 3 + \log_\Delta(n)$, 
where $\Delta = (1+\sqrt{7})^{1/5}$. Furthermore, if we replace $\Delta$ by any real number $\Delta' > \Delta$,
the inequality fails for all large enough~$n$.
\end{theorem}

\begin{proof}[Proof ideas]
The first part of Theorem~\ref{thm:good-size} is proved as follows.
Ideally, we would like to show that
$r_{i+j} \geqslant \Delta^j r_i$
for all $i \geqslant 3$ and some fixed integer $j$. However, these inequalities do not hold for all $i$.
Yet, we prove that they hold if $i+2$ and $i+j+2$ are not obstruction indices, and $i+j+1$ is an obstruction index,
and it follows quickly that $r_h \geqslant \Delta^{h-3}$.

The optimality of $\Delta$ is much more difficult to prove.
It turns out that the constants $2+\sqrt{7}$, $(1+\sqrt{7})^{1/5}$, and the expansion function
referred to in the proof of Proposition~\ref{pro:domination} were constructed as least fixed points of non-decreasing operators,
although this construction needed not be explicit for using these constants and function.
Hence, we prove that $\Delta$ is optimal by inductively constructing sequences of run sizes that show that
$\limsup\{\log(r_h) / h\} \geqslant \Delta$; much care is required for proving that our constructions
are indeed feasible.
\end{proof}


\section{Conclusion}\label{sec:conclu}

At first, when we learned that Java's QuickSort had been replaced by a variant of \MS, we thought that this new algorithm --~\TS~-- should be really fast and efficient in practice, and that we should look into its average complexity to confirm this from a theoretical point of view. 
Then, we realized that its worst-case complexity had not been formally established yet and we first focused on giving a proof that it runs in $\O(n\log n)$, which we wrote in a preprint~\cite{AuNiPi15}. 
In the present article, we simplify this preliminary work and provide a short, simple and self-contained proof of \TS's complexity, which sheds some light on the behavior of the algorithm. 
Based on this description, we were also able to answer positively a natural question, which was left open so far: does \TS runs in $\O(n+n\log\rho)$, where $\rho$ is the number of runs? 
We hope our theoretical work highlights that \TS is actually a very good sorting algorithm. 
Even if all its fine-tuned heuristics are removed, the dynamics of its merges, induced by a small number of local rules, results in a very efficient global behavior, particularly well suited for \emph{almost sorted} inputs.

Besides, we want to stress the need for a thorough algorithm analysis, in order to prevent errors and misunderstandings. As obvious as it may sound, the three consecutive mistakes on the stack height in Java illustrate perfectly how the best ideas can be spoiled by the lack of a proper complexity analysis. 

Finally, following~\cite{GoRoBoBuHa15}, we would like to emphasize that there seems to be no reason not to use the recent version of \TS, which is efficient in practice, formally certified and whose optimal complexity is easy to understand.

\bibliography{timsort_biblio}


\newpage
\appendix

\section{Appendix}

\subsection{Proofs}

We provide below complete proofs of the results mentioned
in Section~\ref{sec:java}.

In what follows, we will often refer to so-called \emph{stable} stacks:
we say that a stack $\S = (R_1,\ldots,R_h)$ is \emph{stable}
if $r_1 + r_2 < r_3$ and $r_1 < r_2$, i.e. if the next operation
that will be performed by \TS is a push operation (Case \#1).

\subsubsection{Proving Proposition~\ref{pro:domination}}\label{sec:domination}

Aiming to prove Proposition~\ref{pro:domination}, and keeping in mind that studying
stable stacks may be easier than studying all stacks,
a first step is to introduce the following quantities.

\begin{definition}\label{def:alphabeta}
Let $n$ be a positive integer.
We denote by $\alpha_n$ (resp., $\beta_n$),
the smallest real number $m$ such that, in every stack (resp., stable stack) $\S = (R_1,\ldots,R_h)$
obtained during an execution of \TS,
and for every integer $i \in \{1,\ldots,h\}$ such that $r_i = n$, we have $r_2+\ldots+r_{i-1} \leqslant m \times r_i$;
if no such real number exists, we simply set $\alpha_n = +\infty$ (resp.,~$\beta_n = +\infty$).
\end{definition}

By construction, $\alpha_n \geqslant \beta_n$ for all $n \geqslant 1$. The following lemma proves that $\alpha_n \leqslant \beta_n$.

\begin{lemma}\label{lem:inv-K}
At any time during the main loop of \TS, if the stack is
$(R_1,\ldots, R_h)$, then we have
(a) $r_i < r_{i+1}$ for all $i \in \{3,4,\ldots,h-1\}$ and
(b) $r_2+\ldots+r_{i-1} \leqslant \beta_n r_i$ for all $n \geqslant 1$ and $i \leqslant h$ such that $r_i = n$.
\end{lemma}

\begin{proof}
Assume that (a) and (b) do not always hold, and consider the first moment where some of them do not hold.
When the main loop starts, both (a) and (b) are true.
Hence, from a stack $\S = (R_1,\ldots,R_h)$, on which (a) and (b) hold,
we carried either a push step (Case~\#1)
or a merging step (Cases~\#2 to \#4),
thereby obtaining the new stack $\overline{\S} = (\overline{R}_1,\ldots,\overline{R}_{\overline{h}})$.
We consider separately these two cases:

\begin{itemize}
\item After a push step, we have $\overline{h} = h+1$ , $r_1 + r_2 < r_3$
(otherwise, we would have performed a merging step instead of a push step)
and $\overline{r}_i = r_{i-1}$ for all $i \geqslant 2$. It follows that
$\overline{r}_3 = r_2 < r_1 + r_2 < r_3 = \overline{r}_4$, and that
$\overline{r}_i = r_{i-1} < r_i = \overline{r}_{i+1}$ for all $i \geqslant 4$.
This proves that $\overline{\S}$ satisfies (a).

In addition, the value of $\overline{r}_1$ has no impact on whether $\overline{\S}$ satisfies (b).
Hence, we may assume without loss of generality that $\overline{r}_1 < \min\{\overline{r}_2,\overline{r}_3-\overline{r}_2\}$
(up to doubling the size of every run ever pushed onto the stack so far and setting $\overline{r}_1 = 1$),
thereby making $\overline{\S}$ stable. This proves that $\overline{\S}$ satisfies (b).

\item After a merging step, we have $\overline{h} = h-1$, $\overline{r}_2 \leqslant r_2 + r_3$
and $\overline{r}_i = r_{i+1}$ for all $i \geqslant 3$.
Hence, $\overline{r}_i = r_{i+1} < r_{i+2} = \overline{r}_{i+1}$ for all $i \geqslant 3$, and $\overline{\S}$ satisfies (a).
Furthermore, we have $0 \leqslant \beta_{\overline{r}_2} \overline{r}_2$, and
$\overline{r}_2 + \overline{r}_3 + \ldots + \overline{r}_i \leqslant r_2 + r_3 + \ldots + r_{i+1} \leqslant \beta_n r_{i+2} = \beta_n \overline{r}_{i+1}$
whenever $i \geqslant 1$ and $\overline{r}_{i+1} = r_{i+2} = n$.
This proves that $\overline{\S}$ also satisfies (b).
\end{itemize}

Hence, in both cases, (a) and (b) also hold in $\overline{\S}$, which contradicts our assumption and completes the proof.
\end{proof}

\begin{corollary}
For all integers $n \geqslant 1$, we have $\alpha_n = \beta_n$.
\end{corollary}

It remains to prove that $\alpha_n \leqslant \alpha_\infty$ for all $n \geqslant 1$,
where $\alpha_\infty = 2 + \sqrt{7}$.
This is the object of the next results.

What makes Java's \TS much harder to study than Python's \TS is the fact that,
during the execution of Java's \TS algorithm, we may have stacks
$\S = (R_1,\ldots,R_h)$ on which the invariant~\eqref{eq:inv1} : $r_i > r_{i-1} + r_{i-2}$
fails for many integers $i \geqslant 3$, possibly consecutive.
In Section~\ref{sec:java}, such integers were called \emph{obstruction indices} of the stack $\S$.
Hence, we focus on sequences of consecutive obstruction indices.

\begin{lemma}\label{lemma:hard-b}
Let $\S = (R_1,\ldots,R_h)$ be a stable stack obtained during the main loop of Java's \TS.
Assume that $i-k,i+1-k,\ldots,i$ are consecutive obstruction indices of $\S$,
and that $\alpha_n \leqslant \alpha_\infty$ for all $n \leqslant r_i - 1$. Then,
\[r_{i-k-2} \leqslant \frac{\alpha_\infty + 1 - k}{\alpha_\infty + 2} r_{i-1}.\]
\end{lemma}

\begin{proof}
Let $T$ be the number of merge or push operations performed between the start of the main loop and the creation of the stack $\S$.
For all $k \in \{0,\ldots,T\}$ and all $j \geqslant 1$, we denote by~$\S_k$ the stack after
$k$ operations have been performed. We also denote by $P_{j,k}$ the $j$\textsuperscript{th} bottom-most
run of the stack $S_k$, and by $p_{j,k}$ the size of $P_{j,k}$; we set $P_{j,k} = \emptyset$ and
$p_{j,k} = 0$ if $S_k$ has fewer than $j$ runs.
Finally, for all $j \leqslant h$, we set
$t_j = \min\{k \geqslant 0 \mid \forall \ell \in \{k,\ldots,T\}, p_{j,\ell} = p_{j,T}\}$.

First, observe that $t_j < t_{j+2}$ for all $j \leqslant h-2$, 
because a run can be pushed or merged only in top or $2$\textsuperscript{nd}-to-top position.
Second, if $t_j \geqslant t_{j+1}$ for some $j \leqslant h-1$,
then the runs $P_{j,t_j}$, $P_{j+1,t_j}$ are the two top runs of $\S_{t_j}$.
Since none of the runs $P_1,\ldots,P_{j+1}$ is modified afterwards,
it follows, if $j \geqslant 2$, that $p_{j+1} + p_{j} = p_{j+1,t_j} + p_{j,t_j} < p_{j-1,t_j} = p_{j-1}$,
and therefore that $h + 2 - j$ is not an obstruction index.

Conversely, let $m_0 = h + 3 - i$.
We just proved that $t_{m_0-2} < t_{m_0}$ and also that $t_{m_0-1} < t_{m_0} < \ldots < t_{m_0 + k}$.
Besides, for all $m \in \{m_0,\ldots,m_0 + k\}$, we prove that the~$t_m$\textsuperscript{th} operation
was a merge operation of type \#2. Indeed, if not, then the run $P_{m,t_m}$ would be the topmost run
of $\S_{t_m}$; since the runs $P_{m-1}$ and $P_{m-2}$ were not modified after that, we would have
$p_m + p_{m-1} < p_{m-2}$, contradicting the fact that $h + 3 - m$ is an obstruction index.
In particular, it follows that $p_{m+1,t_m} = p_{m+2,t_m-1} \geqslant p_{m,t_m-1}$ and that
$p_m = p_{m,t_m} \leqslant p_{m-1,t_m} - p_{m+1,t_m} = p_{m-1} - p_{m+1,t_m}$.

Moreover, for $m = m_0$, observe that $p_m = p_{m,t_m} = p_{m,t_m-1} + p_{m+1,t_m-1}$.
Applying Lemma~\ref{lem:inv-K} on the stacks $\S_T$ and $\S_{t_m-1}$,
we know that $p_{m,t_m-1} \leqslant p_m \leqslant p_{m-2} - 1 = r_i - 1$ and that
$p_{p+1,t_m-1} \leqslant a_{p_{m,t_m-1}} p_{m,t_m-1} \leqslant \alpha_\infty p_{m,t_m-1}$,
which proves that $p_m \leqslant (\alpha_\infty+1) p_{m,t_m-1} \leqslant (\alpha_\infty+1) p_{m+1,t_m}$, i.e.,
$p_{m_0} \leqslant (\alpha_\infty+1) p_{m_0+1,t_{m_0}}$. Henceforth, we set $\kappa = p_{m_0+1,t_{m_0}}$.  

In addition, for all $m \in \{m_0+1,\ldots,m_0 + k\}$, observe that the sequence $(p_{m+1,k})_{t_m \leqslant k \leqslant T}$ is non-decreasing.
Indeed, when $t_m \leqslant k$, and therefore $t_i \leqslant k$ for all $i \leqslant m$, the run $p_{m+1,k}$ can only be modified
by being merged with another run, thereby increasing in size.
This proves that $p_{m+2,t_{m+1}} \geqslant p_{m+1,t_{m+1}-1} \geqslant p_{m+1,t_m}$.
Hence, an immediate induction shows that $p_{m+1,t_m} \geqslant p_{m_0+1,t_{m_0}} = \kappa$ for all $m \in \{m_0,\ldots,m_0+k\}$,
and it follows that $p_m \leqslant p_{m-1} - \kappa$.

Overall, this implies that $r_{i-k-2} = p_{m_0+k} \leqslant p_{m_0} - k \kappa$.
Note that $p_{m_0} \leqslant \min\{(\alpha_\infty+1) \kappa, p_{m_0-1} - p_{m_0+1,t_{m_0}}\} = \min\{(\alpha_\infty+1) \kappa, p_{m_0-1} - \kappa\}$.
It follows that
\[r_{i-k-2} \leqslant \min\{(\alpha_\infty+1) \kappa, p_{m_0-1} - \kappa\} - k \kappa \leqslant \min\{(\alpha_\infty+1-k) \kappa, r_{i-1} - (k+1)\kappa\},\]
whence $(\alpha_\infty+2) r_{i-k-2} \leqslant (k+1) (\alpha_\infty+1-k) \kappa + (\alpha_\infty+1-k) (r_{i-1} - (k+1)\kappa) = (\alpha_\infty+1-k) r_{i-1}$.
\end{proof}

Lemma~\ref{lemma:hard-b} paves the way towards a proof by induction that $\alpha_n \leqslant \alpha_\infty$.
Indeed, a first, immediate consequence of Lemma~\ref{lemma:hard-b}, is that, provided that $\alpha_n \leqslant \alpha_\infty$ for all $n \leqslant r_i-1$,
then the top-most part $(R_1,\ldots,R_i)$ may not contain more than $\alpha_\infty + 2$ (and therefore no more than $6$) consecutive obstruction indices.
This suggests that the sequence $r_1,\ldots,r_i$ should grow ``fast enough'', which might then be used to prove that $\alpha_{r_i} \leqslant \alpha_\infty$.
We present below this inductive proof, which relies on the following objects.

\begin{definition}\label{def:expansion}
We call \emph{expansion function} the function $f : [0,1] \to \mathbb{R}_{\geqslant 0}$ defined by
\[f : x \to \begin{cases} (1+\alpha_\infty) x & \text{if $0 \leqslant x \leqslant 1/2$} \\
x + \alpha_\infty(1-x) & \text{if $1/2 \leqslant x \leqslant \alpha_\infty/(2\alpha_\infty-1)$} \\
\alpha_\infty x & \text{if $\alpha_\infty/(2\alpha_\infty-1) \leqslant x \leqslant 1$.}\end{cases}\]
\end{definition}

In the following, we denote by $\theta$ the real number $\alpha_\infty/(2\alpha_\infty-1)$.
Let us first prove two technical results about the expansion function.

\begin{lemma}\label{lem:f}
We have $\alpha_\infty x \leqslant f(x)$ for all $x \in [0,1]$,
$f(x) \leqslant f(1/2)$ for all $x \in [0,\theta]$,
$f(x) \leqslant f(1)$ for all $x \in [0,1]$,
$x + \alpha_\infty(1-x) \leqslant f(x)$ for all $x \in [1/2,1]$ and
$x + \alpha_\infty(1-x) \leqslant f(1/2)$ for all $x \in [1/2,1]$.
\end{lemma}

\begin{proof}
Since $f$ is piecewise linear, it is enough to check the above inequalities
when $x$ is equal to $0$, $1/2$, $\theta$ or $1$, which is immediate.
\end{proof}

\begin{lemma}\label{lem:f2}
For all real numbers $x, y \in [0,1]$ such that $x (y+1) \leqslant 1$, we have $x (1 + f(y)) \leqslant \min\{f(1/2),f(x)\}$.
\end{lemma}

\vfill

\begin{proof}
We treat three cases separately, relying in each case on Lemma~\ref{lem:f}:
\begin{itemize}
\item if $0 \leqslant x \leqslant 1/2$, then $x (1+f(y)) \leqslant x(1+f(1)) = (1+\alpha_\infty) x = f(x) \leqslant f(1/2)$;
\item if $1/2 < x \leqslant 1$ and $f(1/2) < f(y)$, then $\theta \leqslant y \leqslant 1$, hence
$x (1+f(y)) = x + \alpha_\infty x y \leqslant x + \alpha_\infty (1-x) \leqslant \min\{f(x),f(1/2)\}$;
\item if $0 \leqslant f(y) \leqslant f(1/2)$, and since $\alpha_\infty \geqslant 1$,
we have $x (1 + f(y)) \leqslant x (1+f(1/2)) = x(3 + \alpha_\infty)/2 \leqslant x(1+\alpha_\infty) \leqslant f(x)$;
if, furthermore, $y \leqslant 1/2$, then
\begin{align*}
  x (1+f(y)) & \leqslant (1+(1+\alpha_\infty) y) / (1+y) = (1+\alpha_\infty) - \alpha_\infty / (1+y) \\
  & \leqslant (1+\alpha_\infty) - 2\alpha_\infty/3 = (3+\alpha_\infty)/3,
\end{align*}
and if $1/2 \leqslant y$, then $x (1+f(y)) \leqslant (1+f(1/2)) / (1+y) \leqslant 2(1+f(1/2))/3 = (3+\alpha_\infty)/3$; since $\alpha_\infty \geqslant 3$, it follows that
$x (1+f(y)) \leqslant (3+\alpha_\infty)/3 \leqslant (1+\alpha_\infty) / 2 = f(1/2)$ in both cases.
\end{itemize}
\end{proof}

Using Lemma~\ref{lemma:hard-b} and the above results about the expansion function,
we finally get the following result, of which Proposition~\ref{pro:domination} is an immediate consequence.

\begin{lemma}\label{lem:4b}
Let $\S = (R_1,\ldots,R_h)$ be a stable stack obtained during the main loop of Java's \TS.
For all integers $i \geqslant 2$, we have $r_1+r_2+\ldots+r_{i-1} \leqslant r_i f(r_{i-1} / r_i)$,
where $f$ is the expansion function.
In particular, we have $\alpha_n = \beta_n \leqslant \alpha_\infty$ for all integers $n \geqslant 1$.
\end{lemma}

\begin{proof}
Lemma~\ref{lem:f} proves that $2 x \leqslant \alpha_\infty x \leqslant y f(x/y)$ whenever $0 < x \leqslant y$,
and therefore the statement of Lemma~\ref{lem:4b} is immediate when $i \leqslant 3$.
Hence, we prove Lemma~\ref{lem:4b} for $i \geqslant 4$, and proceed by induction on $r_i = n$, thereby assuming that $\alpha_{n-1}$ exists.

Let $x = r_{i-1} / r_i$ and $y = r_{i-2} / r_{i-1}$. By Lemma~\ref{lem:inv-K}, and since the stack $\S$ is stable,
we know that $r_{i-2} < r_{i-1} < r_i$, and therefore that $x < 1$ and $y < 1$.
If $i$ is not an obstruction index, then we have $r_{i-2} + r_{i-1} \leqslant r_i$, i.e., $x (1+y) \leqslant 1$ and, using Lemma~\ref{lem:f2}, it follows that
$r_1+\ldots+r_{i-1} = (r_1 + \ldots + r_{i-2}) + r_{i-1} \leqslant f(y) r_{i-1} + r_{i-1} = x (1 + f(y)) r_i \leqslant f(x) r_i$.

On the contrary, if $i$ is an obstruction index,
let $k$ be the smallest positive integer such that $i-k$ is not an obstruction index.
Since the stack $\S$ is stable, we have $r_1 + r_2 < r_3$, which means that $3$ is not an obstruction index,
and therefore $i-k \geqslant 3$. Let $u = r_{i-k-1} / r_{i-k}$ and $v = r_{i-k-2} / r_{i-k-1}$. By construction,
we have $r_{i-k-2} + r_{i-k-1} \leqslant r_{i-k}$, i.e., $u(1+v) \leqslant 1$.
Using Lemma~\ref{lemma:hard-b}, and since $r_{i-k-1} < r_i$ and $\alpha_{r_i-1} \leqslant f(1) = \alpha_\infty$ by induction hypothesis, we have
\begin{align*}
r_1+\ldots+r_{i-1} & = (r_1 + \ldots + r_{i-k-2}) + r_{i-k-1} + \ldots + r_{i-1} \leqslant r_{i-k-1} f(v) + r_{i-k-1} + \ldots + r_{i-1} \\
& \leqslant r_{i-k} u (1 + f(v)) + r_{i-k} + \ldots + r_{i-1} \leqslant r_{i-k} f(1/2) + r_{i-k} + \ldots + r_{i-1} \\
& \leqslant \frac{1}{\alpha_\infty+2} \left((\alpha_\infty + 3-k) f(1/2) + \sum_{j=1}^k (\alpha_\infty+3-j) \right) r_{i-1} \\
& \leqslant \frac{1}{2(\alpha_\infty+2)} \left(\alpha_\infty^2 + (4+k)\alpha_\infty - k^2 + 4k + 3\right)r_{i-1}.
\end{align*}

The function $g : t \to \alpha_\infty^2 + (4+t)\alpha_\infty - t^2 + 4t + 3$ takes its maximal value, on the real line,
at $t = (\alpha_\infty+4)/2 \in (4,5)$. Consequently, for all integers $k$, and since $\alpha_\infty \leqslant 5$, we have
\[g(k) \leqslant \max\{g(4),g(5)\} = \alpha_\infty^2 + \max\{8\alpha_\infty+3,9\alpha_\infty-2\} = \alpha_\infty^2 + 8 \alpha_\infty + 3.\]
Then, observe that $2(\alpha_\infty+2)\alpha_\infty = 30 + 12 \sqrt{7} = \alpha_\infty^2 + 8 \alpha_\infty + 3$.
It follows that
\[r_1+\ldots+r_{i-1} \leqslant \frac{\alpha_\infty^2 + 8 \alpha_\infty + 3}{2(\alpha_\infty+2)} r_{i-1} = \alpha_\infty x r_i \leqslant f(x) r_i.\]

Hence, regardless of whether $i$ is an obstruction index or not,
we have $r_1+\ldots+r_{i-1} \leqslant f(x) r_i \leqslant f(1) r_i = \alpha_\infty r_i$, which completes the proof.
\end{proof}

\subsubsection{Proving the first part of Theorem~\ref{thm:good-size}}\label{sec:good1}

We prove below the inequality of Theorem~\ref{thm:good-size};
proving that that the constant $\Delta$ used in Theorem~\ref{thm:good-size} is optimal will be the done in the next section.

In order to carry out this proof, we need to consider some integers of considerable interest.
Let $\S = (R_1,\ldots,R_h)$ be a \emph{stable} stack of runs.
We say that an integer $i \geqslant 1$ is a \emph{growth index} if $i+2$ is \emph{not} an obstruction index,
and that $i$ is a \emph{strong growth index} if $i$ is a growth index and if, in addition, $i+1$ is an obstruction index.
Note that $h$ an $h-1$ are necessarily growth indices, since $h+1$ and $h+2$ are too large to be obstruction indices.

Our aim is now to prove inequalities of the form $r_{i+j} \geqslant \Delta^j r_i$, where $3 \leqslant i \leqslant i+j \leqslant h$.
However, such inequalities do not hold in general, hence we restrict the scope of the integers~$i$ and $i+j$,
which is the subject of the two following results.

\begin{lemma}\label{lem:non-obstruction-growth}
Let $i$ and $j$ be positive integers such that $i+2 \leqslant j \leqslant h$.
If no obstruction index $k$ exists such that $i+2 \leqslant k \leqslant j$, then
$2 \Delta^{j-i-2} r_i \leqslant r_j$.
\end{lemma}

\begin{proof}
For all $n \geqslant 0$, let $F_n$ denote the $n$\textsuperscript{th} Fibonacci number,
defined by $F_0 = 0$, $F_1 = 1$ and $F_{n+2} = F_n + F_{n+1}$ or, alternatively,
by $F_n = (\phi^n - (-\phi)^{-n})/\sqrt{5}$, where $\phi = (1+\sqrt{5})/2$ is the Golden ratio.
Observe now that
\[F_{j-i+1} r_i \leqslant F_{j-i-1} r_i + F_{j-i} r_{i+1} \leqslant F_{j-i-2} r_{i+1} + F_{j-i-1} r_{i+2} 
\leqslant \ldots \leqslant F_0 r_{j-1} + F_1 r_j = r_j.\]

Moreover, for all $n \geqslant 3$, we have $F_n = 2 F_n / F_3 = 2 \phi^{n-3} (1-(-1)^n \phi^{-2n})/(1-\phi^{-6}) \geqslant 2 \phi^{n-3}$.
Since $\Delta < \phi$, it follows that $2 \Delta^{j-i-2} r_i \leqslant F_{j-i+1} r_i \leqslant r_j$.
\end{proof}

\begin{lemma}\label{lem:next-growth}
Let $i$ and $j$ be positive integers such that $1 \leqslant i \leqslant j \leqslant h$.
If $i$ is a growth index and $j$ is a strong growth index, then $\Delta^{j-i} r_i \leqslant r_j$.
\end{lemma}

\begin{proof}
Without loss of generality, let us assume that $i < j$ and that there is no strong growth index $k$ such that $i < k < j$.
Indeed, if such an index $k$ exists, then a simple induction completes the proof of Lemma~\ref{lem:next-growth}.

Let $\ell$ be the largest integer such that $\ell \leqslant j$ and $\ell$ is not an obstruction index.
Lemmas~\ref{lemma:hard-b} and~\ref{lem:4b} prove that $(\alpha_\infty+2)r_\ell \leqslant (\alpha_\infty+2+\ell-j) r_j$
and that $(\alpha_\infty+2)r_{\ell-1} \leqslant (\alpha_\infty+1+\ell-j) r_j$.
The latter inequality proves that $j-\ell \leqslant \lfloor \alpha_\infty+1 \rfloor = 5$.

By construction, we have $i+2 \leqslant \ell$, and no integer $k$ such that $i+2 \leqslant k \leqslant \ell$ is an obstruction index.
Hence, Lemma~\ref{lem:non-obstruction-growth} proves that $2 (\alpha_\infty+2) \Delta^{\ell-i-2} r_i \leqslant(\alpha_\infty+2)r_\ell \leqslant (\alpha_\infty+2+\ell-j) r_j$.
Moreover, simple numerical computations, for $j-\ell \in \{0,\ldots,5\}$, prove that $\Delta^{j-\ell+2} \leqslant 2 (\alpha_\infty+2)/(\alpha_\infty+2+\ell-j)$,
with equality when $j-\ell = 3$.
It follows that $\Delta^{j-i} r_i = \Delta^{j-\ell+2} \Delta^{\ell-i-2} r_i \leqslant r_j$.
\end{proof}

Finally, the inequality of Theorem~\ref{thm:good-size} is an immediate consequence of the following result.

\begin{lemma}
Let $\S = (R_1,\ldots,R_h)$ be a stack obtained during the main loop of Java's \TS.
We have $r_h \geqslant \Delta^{h-3}$.
\end{lemma}

\begin{proof}
Let us first assume that $\S$ is stable.
Then, $r_1 \geqslant 1$, and $1$ is a growth index.
If there is no obstruction index, then Lemma~\ref{lem:non-obstruction-growth} proves that $r_h \geqslant 2 \Delta^{h-3} r_1 \geqslant \Delta^{h-2}$.

Otherwise, let $\ell$ be the largest obstruction index. Then, $\ell-1$ is a strong growth index, and Lemma~\ref{lem:next-growth} proves that
$r_{\ell-1} \geqslant \Delta^{\ell-2} r_1 \geqslant \Delta^{\ell-2}$. If $\ell = h$, then $r_h \geqslant r_{\ell-1} \geqslant \Delta^{h-2}$,
and if $\ell \leqslant h-1$, then Lemma~\ref{lem:non-obstruction-growth} also
proves that $r_h \geqslant 2 \Delta^{h-\ell-1} r_{\ell-1} \geqslant \Delta^{h-\ell} \Delta^{\ell-2} = \Delta^{h-2}$.

Finally, if $\S$ is not stable, the result is immediate for $h \leqslant 3$, hence we assume that $h \geqslant 4$.
The stack $\S$ was obtained by pushing a run onto a stable stack $\S'$ of size at least $h-1$, then merging runs from
$\S'$ into the runs $R_1$ and $R_2$. It follows that $R_h$ was already the largest run of $S'$, and therefore that $R_h \geqslant \Delta^{h-3}$. 
\end{proof}

\subsubsection{Proving the second part of Theorem~\ref{thm:good-size}}\label{sec:good2}

We finally focus on proving that the constant $\Delta$ of Theorem~\ref{thm:good-size} is optimal.
The most important step towards this result consists in proving that $\alpha_\infty = \lim_{n \to \infty} \alpha_n$,
with the real numbers $\alpha_n$ introduced in Definition~\ref{def:alphabeta} and $\alpha_\infty = 2 + \sqrt{7}$.
Since it is already proved, in Lemma~\ref{lem:4b}, that $\alpha_n \leqslant \alpha_\infty$ for all $n \geqslant 1$,
it remains to prove that $\alpha_\infty \leqslant \liminf_{n \to \infty} \alpha_n$.
We obtain this inequality by constructing explicitly, for $k$ large enough,
a stable sequence of runs $(R_1,\ldots,R_h)$ such that $r_h = k$ and $r_2 + \ldots + r_{h-1} \approx \alpha_\infty k$.
Hence, we focus on constructing sequences of runs.

In addition, let us consider the stacks of runs created
by the main loop of Java's \TS on a sequence of runs $P_1,\ldots,P_n$.
We say below that the sequence $P_1,\ldots,P_k$
\emph{produces} a stack of runs $\S = (R_1,\ldots,R_h)$ if the stack $\S$ is obtained
after each of the runs $P_1,\ldots,P_n$ has been pushed;
observe that the sequence $P_1,\ldots,P_n$ produces exactly one stable stack.
We also say that a stack of runs is \emph{producible} if it is produced by some sequence of runs.

Finally, in what follows, we are only concerned with run sizes.
Hence, we abusively identify runs with their sizes.
For instance, in Figure~\ref{fig:inv_still_broken},
the sequence $(109, 83, 25, 16, 8, 7, 26, 2, 27)$ produces the stacks
$(27,2,26,56,83,109)$ and $(27,28,56,83,109)$;
only the latter stack is stable.

We review now several results that will provide us with convenient and powerful ways of constructing producible stacks.

\begin{lemma}\label{lem:6.5}
Let $\S = (r_1,\ldots,r_h)$ be a stable stack
produced by a sequence of runs $p_1,\ldots,p_n$.
Assume that $n$ is minimal among all sequences that produce $\S$.
Then, when producing $\S$, no merging step \#3 or \#4 was performed.

Moreover, for all $k \leqslant n-1$,
after the run $p_{k+1}$ has been pushed onto the stack,
the elements coming from $p_k$ will never belong to the topmost run of the stack.
\end{lemma}

\begin{proof}
We begin by proving the first statement of Lemma~\ref{lem:6.5} by induction on $n$,
which is immediate for $n = 1$.
Hence, we assume that $n \geqslant 2$, and we further assume, for the sake of contradiction, that some merging step \#3 or \#4 took place.
Let $\S' = (r'_1,\ldots,r'_\ell)$ be the stable stack produced by the sequence $p_1,\ldots,p_{n-1}$.
By construction, this sequence is as short as possible, and therefore no merging step \#3 or \#4 was used so far.
Hence, consider the last merging step \#3 or \#4, which necessarily appears after $p_n$ was pushed onto the stack.
Just after this step has occurred, we have a stack $\S'' = (r''_1,\ldots,r''_m)$,
with $r'_i = r''_j$ whenever $j \geqslant 2$ and $i+m = j + \ell$,
and the run $r''_1$ was obtained by merging the runs $p_n, r'_1,\ldots,r'_{\ell+1-m}$.

Let $p_1,\ldots,p_k$ be the runs that had been pushed onto the stack when the run $r''_2 = r'_{m+2-\ell}$ was created.
This creation was the result of either a push step or a merging step \#2.
In both cases, and until $\S'$ is created, no merging step \#3 or \#4
ever involves any element placed within or below $r''_2$.
Then, in the case of a push step, we have $p_k = r''_2$, and therefore
the sequence $\mathcal{P}_{\#1} = (p_1,\ldots,p_k,r''_1)$ also produces the stack $\S'$.
In the case of a merging step~\#2, it is the sequence $\mathcal{P}_{\#2} = (p_1,\ldots,p_{k-1},r''')$ that also produces the stack $\S''$,
where $r'''$ is obtained by merging the runs $p_k,\ldots,p_n$.

In both cases, since the sequences $\mathcal{P}_{\#1}$ and $\mathcal{P}_{\#2}$ produce $\S''$,
they also produce $\S$. It follows that $k+1 \geqslant n$ (in the first case) or $k \geqslant n$ (in the second case).
Moreover, the run~$r''_2$ was not modified between pushing the run $p_n$ and before obtaining the stack $\S''$,
hence $k \leqslant n-1$. This proves that $k = n-1$ and that the run $r''_2$ was obtained through a push step,
i.e. $p_{n-1} = r''_2$. But then, the run $r''_1$ may contain elements of $p_n$ only, hence
is not the result of a merging step \#3 or \#4: this disproves our assumption and proves the first statement of Lemma~\ref{lem:6.5}.

Finally, observe that push steps and merging steps \#2 never allow a run in $2$\textsuperscript{nd}-to-top position or below
to go to the top position. This completes the proof of Lemma~\ref{lem:6.5}.
\end{proof}

\begin{lemma}\label{lem:6.6}
Let $\S = (r_1,\ldots,r_h)$ be a stable stack
produced by a sequence of runs $p_1,\ldots,p_n$.
Assume that $n$ is minimal among all sequences that produce $\S$.
There exist integers $i_0,\ldots,i_h$ such that $0 = i_h < i_{h-1} < \ldots < i_0 = n$ and such that, for every integer $k \leqslant h$,
(a) the runs $p_{i_k+1},\ldots,p_{i_{k-1}}$ were merged into the run $r_k$, and
(b) $i_{k-1} = i_k+1$ if and only if $k+2$ is not an obstruction index.
\end{lemma}

\vfill

\begin{proof}
The existence (and uniqueness) of integers $i_0,\ldots,i_h$ satisfying (a) is straightforward, hence we focus on proving (b).
That property is immediate if $h = 1$, hence we assume that $2 \leqslant h \leqslant n$.
Checking that the sequence $r_h,r_{h-1},p_{i_{h-2}+1},p_{i_{h-2}+2},\ldots,p_n$ produces the stack~$\S$ is immediate,
and therefore $i_{h-1} = 1$ and $i_{h-2} = 2$, i.e., $r_h = p_1$ and $r_{h-1} = p_2$.

Consider now some integer $k \leqslant h-2$, and let $\S'$ be the stable stack produced by $p_1,\ldots,p_{i_k+1}$.
From that point on, the run $p_{i_{k-1}+1}$ will never be the topmost run, and the runs
$p_j$ with $j \leqslant i_{k-1}$, which can never be merged together with the run $p_{i_{k-1}+1}$,
will never be modified again. This proves that $\S' = (p_{i_{k-1}+1},r_{k+1},\ldots,r_h)$.

Then, assume that $i_{k-1} = i_k+1$, and therefore that $p_{i_{k-1}+1} = r_k$.
Since $\S'$ is stable, we know that $r_k + r_{k+1} < r_{k+2}$, which means that
$k+2$ is not an obstruction index.
Conversely, if~$k+2$ is not an obstruction index, both sequences
$p_1,\ldots,p_{i_{k+1}+1}$ and $p_1,\ldots,p_{i_{k-1}},r_k,p_{i_{k+1}+1}$ produce the stack
$(p_{i_{k+1}+1},r_k,\ldots,r_h)$ and, since $n$ is minimal, $i_{k-1} = i_k+1$.
\end{proof}

\begin{lemma}\label{lem:6.9}
Let $\S = (r_1,\ldots,r_h)$ be a producible stable stack of height $h \geqslant 3$.
There exists an integer $\kappa \in \{1,4\}$ and a producible stable stack $\S' = (r'_1,\ldots,r'_\ell)$
such that $\ell \geqslant 2$, $r_h = r'_\ell$, $r_{h-1} = r'_{\ell-1}$
and $r_1+\ldots+r_{h-2} = r'_1 + \ldots+  r'_{\ell-2} + \kappa$.
\end{lemma}

\begin{proof}
First, Lemma~\ref{lem:6.9} is immediate if $h = 3$, 
since the sequence of runs $(r_3,r_2,r_1-1)$ produces the stack $(r_1-1,r_2,r_3)$.
Hence, we assume that $h \geqslant 4$.
Let $p_1,\ldots,p_n$ be a sequence of runs, with $n$ minimal, that produces $\S$.
We prove Lemma~\ref{lem:6.9} by induction on $n$.

If the last step carried when producing $\S$ was pushing the run $P_n$ onto the stack,
then the sequence $p_1,\ldots,p_{n-1},p_n-1$ produces the stack $r_1-1,r_2,\ldots,r_h$, and we are done in this case.
Hence, assume that the last step carried was a merging step \#2.

Let $\S' = (q_1,\ldots,q_m)$ be the stable stack produced by the sequence $p_1,\ldots,p_{n-1}$,
and let~$i$ be the largest integer such that $q_i < p_n$.
After pushing $p_n$, the runs $q_1,\ldots,q_i$ are merged into one single run $r_2$,
and we also have $p_n = r_1$ and $q_{i+j} = r_{2+j}$ for all $j \geqslant 1$.
Incidentally, this proves that $m=h+i-2$ and, since $h \geqslant 4$, that $i \leqslant m-2$.
We also have $i \geqslant 2$, otherwise, if~$i = 1$, we would have had a merging step \#3 instead.

If $r_1 = p_n \geqslant q_i + 2$, then the sequence $p_1,\ldots,p_{n-1},p_n-1$ also
produces the stack $(r_1-1,r_2,\ldots,r_h)$, and we are done in this case. Hence, we further assume that $r_1 = p_n = q_i + 1$.
Since $q_{i-1} + q_i \leqslant r_2 < r_3 = q_{i+1}$, we know that $i+1$ is not an obstruction index of $S$.
Let $a \leqslant n-1$ be a positive integer such that $p_1+\ldots+p_a = q_i+q_{i+1}+\ldots+q_m$.
Lemma~\ref{lem:6.6} states that $p_{a+1} = q_{i-1}$,
and therefore that the sequence of runs $p_1,\ldots,p_{a+1}$ produces the stack $(q_{i-1},\ldots,q_m)$.

If $i = 2$ and if $q_1 \geqslant 3$, then $(q_1-1)+q_2 \geqslant q_2+2 > r_1$, and therefore
the sequence of runs $p_1,\ldots,p_a,q_1-1,r_1$ produces
the stable stack $(r_1,r_2-1,r_3,\ldots,r_h)$. However, if~$i = 2$ and $q_1 \leqslant 2$, then
the sequence of runs $p_1,\ldots,p_a,r_1-2$ produces
the stable stack $(r_1-2,r_2-2,\ldots,r_h)$.
Hence, in both cases, we are done, by choosing respectively $\kappa = 1$ and $\kappa = 4$.

Les us now assume that $i \geqslant 3$. Observe that $n \geqslant a + i \geqslant 3$
since, after the stack $(q_{i-1},\ldots,q_m)$ has been created, it remains to create
runs $q_1,\ldots,q_{i-2}$ and finally, $r_1$, which requires pushing at least $i-1$ runs in addition
to the $a+1$ runs already pushed.
Therefore, we must have $q_1 + \ldots + q_{i-1} \geqslant q_i$,
unless what the sequence $p_1,\ldots,p_a,q_1 + \ldots + q_{i-1},p_n$
would have produced the stack $\S$, despite being of length
$a + 2 < n$. In particular, since $q_1 + q_2 < q_3$, it follows that $i \geqslant 4$. Consequently, we have
$q_1 \geqslant 1$, $q_2 \geqslant 2$, $q_3 \geqslant 4$,
and therefore $q_1+\ldots+q_{i-1} \geqslant 7$, i.e. $r_2 \geqslant q_2 + 7 = r_1 + 6$.

Finally, by induction hypothesis, there exists a sequence $p'_1,\ldots,p'_u$, with $u$ minimal,
that produced the stack $(q'_1,\ldots,q'_v)$ such that $q'_v = q_i$, $q'_{v-1} = q_{i-1}$ and
$q_1+\ldots+q_{i-2} = q'_1+\ldots+q'_{v-2} + \kappa$ for some $\kappa \in \{1,4\}$.
Lemma~\ref{lem:6.6} also states that $p'_1 = q_i$ and that $p'_2 = q_{i-1}$.
It is then easy to check that the sequence of runs $p_1,\ldots,p_{b+1},p'_3,\ldots,p'_u$
produces the stable stack $(q'_1,\ldots,q'_{v-2},q_{i-1},q_i,\ldots,q_m)$.
Since $q'_j < q_{i-1} < q_i < r_1 < r_3 = q_{i+1}$ for all $j \leqslant v-2$,
pushing the run $p_n = r_1$ onto that stack and letting merging steps \#2 occur then gives the stack
$(r_1,r_2-\kappa,r_3,\ldots,r_h)$, which completes the proof since $r_1 \leqslant r_2 - 6 < r_2 - \kappa$.
\end{proof}

In what follows, we will only consider stacks that are producible and stable.
Hence, from this point on, we omit mentioning
that they systematically must be producible and stable,
and we will say that ``the stack $\S$ exists'' in order to say that
``the stack $\S$ is producible and stable''.

\begin{lemma}\label{lem:7.0}
Let $\S = (r_1,\ldots,r_h)$ and $\S' = (r'_1,\ldots,r'_\ell)$ be two stacks. Then
(a) for all~$i \leqslant h$, there exists a stack $(r_1,\ldots,r_i)$, and
(b) if $r_{h-1} = r'_1$ and $r_h = r'_2$, then there exists a stack $(r_1,\ldots,r_h,r'_3,\ldots,r'_\ell)$.
\end{lemma}

\begin{proof}
Let $p_1,\ldots,p_m$ and $p'_1,\ldots,p'_n$ be two sequences that respectively produce $\S$ and $\S'$.
Let us further assume that $m$ is minimal.
First, consider some integer $i \leqslant h$, and let $a$ be the integer such that $p_1+\ldots+p_a = r_{i+1}+\ldots+r_h$.
It comes at once that the sequence $p_{a+1},\ldots,p_m$ produces the stack $(r_1,\ldots,r_i)$.
Second, since $m$ is minimal, Lemma~\ref{lem:6.6} proves that $p_1 = r_h = r'_2$ and that $p_2 = r_{h-1} = r'_1$ and,
once again, the sequence $p'_1,\ldots,p'_n,p_3,\ldots,p_3$ produces the stack $(r_1,\ldots,r_h,r'_3,\ldots,r'_\ell)$,
which is also stable.
\end{proof}

\begin{lemma}\label{lem:7}
For all positive integers $k$ and $\ell$ such that $k \leqslant \ell \alpha_\ell$,
there exists a stack $(r_1,\ldots,r_h)$
such that $r_h = \ell$, $k-3 \leqslant r_1+\ldots+r_{h-1} \leqslant k$,
and $r_1+\ldots+r_{h-1} = k$ if $k = \ell\alpha_\ell$.
\end{lemma}

\begin{proof}
First, if $\ell = 1$, then $\alpha_\ell = 0$, and therefore Lemma~\ref{lem:7} is vacuously true.
Hence, we assume that $\ell \geqslant 2$.
Let $\Omega$ be the set of integers $k$ for which
some sequence of runs $p_1,\ldots,p_m$ produces a stack $\S = (r_1,\ldots,r_h)$
such that $r_1+\ldots+r_{h-1} = k$ and $r_h = \ell$.
First, if $k \leqslant \ell-1$, the sequence $\ell,k$ produces the stack $(\ell,k)$,
thereby proving that $\{1,2,\ldots,\ell-1\} \in \Omega$.
Second, it follows from Lemma~\ref{lem:7.0} that $\ell \alpha_\ell \in \Omega$.

Finally, consider some integer $k \in \Omega$ such that $k \geqslant \ell$, and let $p_1,\ldots,p_m$
be a sequence of runs that produces a stack 
$(r_1,\ldots,r_h)$ such that $r_1+\ldots+r_{h-1} = k$ and $r_h = \ell$.
Since $k \geqslant \ell = r_h > r_{h-1}$, we know that $h \geqslant 3$. Hence, due to Lemma~\ref{lem:6.9},
either $k-1$ or $k-4$ (or both) belongs to $\Omega$. This completes the proof of Lemma~\ref{lem:7}.
\end{proof}

\begin{lemma}\label{lem:8}
For all positive integers $k$ and $\ell$ such that $k \leqslant \ell$, we have
$\alpha_\ell \geqslant (1 - k/\ell) \alpha_k$ and $\alpha_\ell \geqslant k \alpha_k / \ell$.
\end{lemma}

\begin{proof}
Let $n = \lfloor \ell / k \rfloor$. Using Lemma~\ref{lem:7}, there exists a sequence of runs $p_1,\ldots,p_m$
that produces a stack $(r_1,\ldots,r_h)$ such that $r_h = k$ and $r_1+\ldots+r_{h-1} = k \alpha_k$.
By choosing~$m$ minimal, Lemma~\ref{lem:6.6} further proves that $p_1 = r_h = k$.
Consequently, the sequence of runs $n p_1,\ldots, n p_{m-1}, \ell$ produces the stack
$(n r_1,\ldots,n r_{h-1},\ell)$, and therefore we have
$\ell \alpha_\ell \geqslant n(r_1+\ldots+r_{h-1}) = n k \alpha_k \geqslant \max\{1,\ell / k - 1\} k \alpha_k = \ell \max\{k/\ell,1 - k / \ell\} \alpha_k$.
\end{proof}

A first intermediate step towards proving that $\lim_{n \to \infty} \alpha_n = \alpha_\infty$
is the following result, which is a consequence of the above Lemmas.

\begin{proposition}\label{pro:8.1}
Let $\overline{\alpha} = \sup\{\alpha_n \mid n \in \mathbb{N}^\ast\}$.
We have $1+\alpha_\infty/2 < \overline{\alpha} \leqslant \alpha_\infty$, and $\alpha_n \to \overline{\alpha}$ when $n \to +\infty$.
\end{proposition}

\begin{proof}
Let $\overline{\alpha} = \sup\{\alpha_n \mid n \in \mathbb{N}^\ast\}$. Lemma~\ref{lem:4b} proves that $\overline{\alpha} \leqslant \alpha_\infty$.
Then, let $\varepsilon$ be a positive real number, and let $k$ be a positive integer such that $\alpha_k \geqslant \overline{\alpha} - \varepsilon$.
Lemma~\ref{lem:8} proves that $\liminf \alpha_n \geqslant \alpha_k \geqslant \overline{\alpha} - \varepsilon$,
and therefore $\liminf \alpha_n \geqslant \overline{\alpha}$.
This proves that $\alpha_n \to \overline{\alpha}$ when $n \to +\infty$.

Finally, it is tedious yet easy to verify that the sequence
$360$, $356$, $3$, $2$, $4$, $6$, $10$, $2$, $1$, $22$, $4$, $2$, $1$, $5$, $1$, $8$, $4$, $2$, $1$, $73$, $4$,
$2$, $5$, $7$, $2$, $16$, $3$, $2$, $4$, $6$, $21$, $4$, $2$, $22$, $4$, $2$, $1$, $5$, $8$, $3$, $2$, $79$,
$3$, $2$, $4$, $6$, $2$, $10$, $6$, $3$, $2$, $33$, $4$, $2$, $5$, $7$, $1$, $13$, $4$, $2$, $1$, $5$, $1$,
$80$, $4$, $2$, $5$, $7$, $1$, $95$, $3$, $2$, $4$, $6$, $10$, $20$, $4$, $2$, $5$, $7$, $3$, $2$, $26$, $6$,
$3$, $1$, $31$, $3$, $2$, $4$, $6$, $2$, $1$, $12$, $4$, $2$, $5$
produces the stack $(5$, $6$, $12$, $18$, $31$, $36$, $68$, $95$, $99$, $195$, $276$, $356$, $360)$. Moreover,
since $(20 \sqrt{7})^2 = 2800 < 2809 = 53^2$, it follows that $80 + 20 \sqrt{7} < 133$, i.e., that $1 + \alpha_\infty/2 = 2 + \sqrt{7}/2 < 133/40$. This proves that
\[\overline{\alpha} \geqslant \alpha_{360} \geqslant \frac{5+6+12+18+31+36+68+95+99+195+276+356}{360} = \frac{133}{40} > 1 + \alpha_\infty/2.\]
\end{proof}

\begin{lemma}\label{lem:8.2}
There exists a positive integer $N$ such that, for all integers $n \geqslant N$ and
$k = \lfloor (n-6)/(\overline{\alpha}+2) \rfloor$,
the stack $(k+1,k+,\alpha_k,n-4,n)$ exists. 
\end{lemma}

\begin{proof}
Proposition~\ref{pro:8.1} proves that there exists a positive real number $\nu > 1+\alpha_\infty/2$
and a positive integer $L \geqslant 256$ such that $\alpha_\ell \geqslant \nu$ for all $\ell \geqslant L$.
Then, we set $N = \lceil (\overline{\alpha}+2)L \rceil + 6$.
Consider now some integer $n \geqslant N$, and let $k = \lfloor (n-6)/(\overline{\alpha}+2) \rfloor$. By construction, we have $k \geqslant L$,
and therefore $\alpha_k \geqslant \nu$.

Let $p_1,\ldots,p_m$ be a sequence of runs that produces a stack $(r_1,\ldots,r_h)$ such that
$r_h = k$ and $r_1+\ldots+r_{h-1} = k \alpha_k$. Lemma~\ref{lem:4b} proves that $\alpha_k \leqslant f(r_{h-1}/k)$,
where $f$ is the expansion function of Definition~\ref{def:expansion}.
Since $\alpha_k \geqslant \nu > 1+\alpha_\infty/2 = f(1/2)$, it follows that $k > r_{h-1} > \theta k$.
Assuming that $m$ is minimal, Lemma~\ref{lem:6.6} proves that $p_1 = r_h = k$ and that $p_2 = r_{h-1}$.

Now, let $k' = \lfloor k / 2 \rfloor + 1$, and let $\ell$ be the largest integer such that $2^{\ell+4} \leqslant k'$.
Since $k \geqslant L \geqslant 256$, we know that $k' \geqslant 128$, and therefore that $\ell \geqslant 3$.
Observe also that, since $\theta = \alpha_\infty/(2\alpha_\infty-1) > 11/20$ and $k' \geqslant 20$, we have
$r_{m-1} > \theta k \geqslant \lfloor k / 2 \rfloor + k/20 \geqslant k'$.
We build a stack of runs $p_2,k,n-4,n$ by distinguishing several cases, according to the value of~$k' / 2^\ell$.

\begin{itemize} 
 \item If $16 \times 2^\ell \leqslant k' \leqslant 24 \times 2^\ell + 1$, let $x$ be
 the smallest integer such that $x \geqslant 2$ and $k' < 2(9 \times 2^\ell + x + 1)$.
 Since $k' \leqslant 24 \times 2^\ell + 1$, we know that $x \leqslant 3 \times 2^\ell$, and that $x = 2$ if $k \leqslant 18 \times 2^\ell + 1$.
 Therefore, the sequence of runs $(n,n-4,3,2,4,6,10,3 \times 8,3 \times 16,\ldots,3 \times 2^\ell,3 \times 2^\ell+x)$ produces the stack
 $(3 \times 2^\ell+x,3 \times 2^{\ell+1}+1,n-4,n)$. Moreover, observe that
 $(3 \times 2^{\ell+1}+1) + (9 \times 2^\ell + x + 1) = 15 \times 2^\ell + 4 < k'$ if $16 \times 2^\ell \leqslant k' \leqslant 18 \times 2^\ell + 1$, and that
 $(3 \times 2^{\ell+1}+1) + (9 \times 2^\ell + x + 1) \leqslant 18 \times 2^\ell + 1 < k'$ if $18 \times 2^\ell + 2 \leqslant k'$.
 Since, in both cases, we also have $k' < 2(9 \times 2^\ell + x + 1)$,
 it follows that $3 \times 2^{\ell+1}+1 < k' - (9 \times 2^\ell + x + 1) < 9 \times 2^\ell + x + 1$.
 
 Consequently, pushing an additional run of size $k'-(9 \times 2^\ell + x + 1)$
 produces the stack $k'-(9 \times 2^\ell + x + 1),9 \times 2^\ell + x + 1,n-4,n$.
 Finally, the inequalities $2k' > k$, $k - k' \geqslant k/2-1 \geqslant k' - 2 \geqslant 18 \times 2^\ell$ and $x \leqslant 3 \times 2^\ell$ prove that
 \[9 \times 2^\ell + x + 1 \leqslant 12 \times 2^\ell + 1 < 16 \times 2^\ell - 2\leqslant k - k' < k'.\]
 Recalling that $k > p_2 > k'$, it follows that pushing additional runs of sizes $k - k'$ and $p_2$
 produces the stack $(p_2,k,n-4,n)$.
 
 \item If $24 \times 2^\ell + 2 \leqslant k' < 32 \times 2^\ell$, let $x$
 be the smallest integer such that $x \geqslant 2$ and $k' < 2(12 \times 2^\ell + x + 1)$.
 Since $k' < 32 \times 2^\ell$, we know that $x+1 \leqslant 2^{\ell+2}$.
 Therefore, the sequence of runs $n,n-4,3,2,4,8,16,32,\ldots,2^{\ell+2},2^{\ell+2}+x$ produces the stack
 $(2^{\ell+2}+x,2^{\ell+3}+1,n-4,n)$. Moreover, the inequalities
 \[(2^{\ell+3}+1) + (3 \times 2^{\ell+2} + x + 1) \leqslant 6 \times 2^{\ell+2} + 1 < k' < 2(3 \times 2^{\ell+2} + x + 1)\]
 prove that $2^{\ell+3}+1 < k'-(3 \times 2^{\ell+2}+x+1) < 3 \times 2^{\ell+2}+x+1$.
 
 Consequently, pushing an additional run of size $k'-(3 \times 2^{\ell+2}+x+1)$
 produces the stack $(k'-(3 \times 2^{\ell+2}+x+1),3 \times 2^{\ell+2}+x+1,n-4,n)$.
 Finally, the inequalities $2k' > k$, $k - k' \geqslant k/2-1 \geqslant k'-2 \geqslant 6 \times 2^{\ell+2}$ and $x+1 \leqslant 2^{\ell+2}$ prove that
 \[3 \times 2^{\ell+2}+x+1 \leqslant 4 \times 2^{\ell+2} < 6 \times 2^{\ell+2} \leqslant k - k' < k'.\]
 Recalling once again that $k > p_2 > k'$, it follows that pushing additional runs of sizes $k - k'$ and $p_2$
 produces the stack $(p_2,k,n-4,n)$ in this case too.
\end{itemize}

Finally, after having obtained the stack $(p_2,k,n-4,n)$, let us add the sequence of runs $p_3,\ldots,p_m,k+1$.
Since $k(1+\alpha_k) + (k+1) \leqslant k(\overline{\alpha}+2)+1 \leqslant n-6+1 < n-4$,
adding these runs produces the stack $(k+1,k + k \alpha_k,n-4,n)$, which completes the proof.
\end{proof}

\begin{lemma}\label{lem:8.3}
For all integers $n \geqslant N$ and
$k = \lfloor (n-6)/(\overline{\alpha}+2) \rfloor$
there exists a stack $(k+3, k(\alpha_k-1)-7-x-y, k\alpha_k-3-x, k(1+\alpha_k), n-4, n)$
for some integers $x,y \in \{0,1,2,3\}$. 
\end{lemma}

\begin{proof}
Consider some integer $n \geqslant N$, and let $k = \lfloor (n-6)/(\alpha_\infty+2) \rfloor$.
By Lemma~\ref{lem:8.2}, there exists a stack $(k+1, k+k\alpha_k,n-4,n)$.

Lemma~\ref{lem:8} proves then that $k \alpha_k \geqslant (k+1) \alpha_{k+1}$. Therefore, Lemma~\ref{lem:7} proves that
there exists a stack $(r_1,\ldots,r_h)$ such that $r_h = k+1$ and $r_1+\ldots+r_{h-1} = k (\alpha_k-1)-4-x$
for some integer $x \in \{0,1,2,3\}$.
By construction, we have $r_1 < r_2 < \ldots < r_h$, hence $r_{h-1} + r_h < 2(k+1) < k(1+\alpha_k)$.
Consequently, there also exists a stack $(r_1,r_2,\ldots,r_{h-1},k+1, k(1+\alpha_k),n-4,n)$.
Then, $k+2 + r_1+\ldots+r_h \leqslant k+2 + k (\alpha_k-1) - 4 + k+1 = k(1+\alpha_k)-1 < k(1+\alpha_k)$,
and therefore pushing an additional run of size $k+2$ produces a stack $(k+2, k\alpha_k-3-x, k(1+\alpha_k), n-4, n)$.

Once again, there exists a stack $(r'_1,\ldots,r'_{h'})$ such that $r'_{h'} = k+2$ and $r'_1+\ldots+r'_{h'-1} = k (\alpha_k-2)-9-x-y$
for some integer $y \in \{0,1,2,3\}$.
By construction, we have $r'_1 < r'_2 < \ldots < r'_{h'}$, hence $r'_{h'-1} + r'_{h'} < 2(k+2) < k\alpha_k-3-x$,
and therefore there exists a stack $(r'_1,r'_2,\ldots,r'_{h'-1},k+2, k\alpha_k-3-x, k(1+\alpha_k), n-4, n)$.
Then, $k+3 + r'_1+\ldots+r'_{h'} \leqslant k+3 + k (\alpha_k-2)-9-x-y + k + 2 = k\alpha_k-4-x-y < k\alpha_k-3-x$,
and therefore pushing an additional run of size $k+3$ produces a stack $(k+3, k(\alpha_k-1)-7-x-y, k\alpha_k-3-x, k(1+\alpha_k), n-4, n)$.
\end{proof}

We introduce now a variant of the real numbers $\alpha_n$ and $\beta_n$,
adapted to our construction.

\begin{definition}\label{def:alphabeta2}
Let $n$ be a positive integer.
We denote by $\gamma_n$ the smallest real number $m$ such that, in every stack $\S = (r_1,\ldots,r_h)$
such that $h \geqslant 2$, $r_h = n$ and $r_{h-1} = n-4$, we have $r_1+\ldots+r_{h-1} \leqslant m \times n$.
If no such real number exists, we simply set $\gamma_n = +\infty$.
\end{definition}

\begin{lemma}\label{lem:8.4}
Let $\underline{\gamma} = \liminf \gamma_n$.
We have $\underline{\gamma} \geqslant (\overline{\alpha} \underline{\gamma} + 9 \overline{\alpha} - \underline{\gamma} + 3)/(2\overline{\alpha}+4)$.
\end{lemma}

\begin{proof}
Let $\varepsilon$ be a positive real number, and let $N_\varepsilon \geqslant N$ be an integer such that $\alpha_\ell \geqslant \overline{\alpha}-\varepsilon$ and
$\gamma_\ell \geqslant \underline{\gamma}-\varepsilon$ for all $\ell \geqslant \lfloor (N-6)/(\alpha_\infty+2) \rfloor$.
Since $\overline{\alpha} > 3$, and up to increasing the value of $N_\varepsilon$,
we may further assume that $\ell(\alpha_\ell-3) \geqslant 30$ for every such integers $\ell$.

Then, consider some integer $n \geqslant N_\varepsilon$, and let $k = \lfloor (n-6)/(\overline{\alpha}+2) \rfloor$.
By Lemma~\ref{lem:8.3}, there exists a stack $(k+3, k(\alpha_k-1)-7-x-y, k\alpha_k-3-x, k(1+\alpha_k), n-4, n)$
for some integers $x,y \in \{0,1,2,3\}$. Let also $k' = \lfloor (k(\alpha_k-1)-7-x-y) / 2 \rfloor$.
Since $k(\alpha_k-3) \geqslant 30$,
it follows that $2 (k' - 4) \geqslant k(\alpha_k-1)-7-x-y - 2 - 8 \geqslant k(\alpha_k-1) - 23 \geqslant 2k+7 > 2(k+3)$.
Similarly, and since $\alpha_k \leqslant \alpha_\infty < 5$, we have
$k' \leqslant k(\alpha_k-1) / 2 < 2k < 2(k+3)$.
Consequently, pushing additional runs of sizes $k'-(k+3)$ and $k'-4$ produces the stack $(k'-4,k',k(\alpha_k-1)-7-x-y, k\alpha_k-3-x, k(1+\alpha_k), n-4, n)$.

Finally, by definition of $\gamma_n$, there exists a sequence of runs $p_1,\ldots,p_m$
that produces a stack $(r_1,\ldots,r_h)$ such that $p_1 = r_h = k'$, $p_2 = r_{h-1} = k'-4$ and
$r_1+\ldots+r_{h-1} = k' \gamma_{k'}$.
Hence, pushing the runs $p_3,\ldots,p_m$ produces the stack
$(r_1,\ldots,r_{h-1},k',k(\alpha_k-1)-7-x-y, k\alpha_k-3-x, k(1+\alpha_k), n-4, n)$.

Then, recall that that $\underline{\gamma} \leqslant \overline{\alpha}$, that $3 < \overline{\alpha} \leqslant \alpha_\infty < 5$ and that $0 \leqslant x,y \leqslant 3$.
It follows that $k \geqslant (n-6)/(\overline{\alpha}+2)-1 \geqslant n/(\overline{\alpha}+2) - 3$ and that
$k' \geqslant (k(\alpha_k-1)-7-x-y) / 2 - 1 \geqslant k(\alpha_k-1)/2 - 8$.
This proves that
\begin{align*}
n \gamma_n & \geqslant r_1+\ldots+r_{h-1} + k' + k(\alpha_k-1)-7-x-y + k\alpha_k-3-x + k(1+\alpha_k) + n-4 \\
& \geqslant k' \gamma_{k'} + k' + k(\alpha_k-1)-13 + k\alpha_k-6+k(1+\alpha_k) + n-4 \\
& \geqslant k' (1 + \underline{\gamma}-\varepsilon) + 3 k \alpha_k + n - 23 \\
& \geqslant (k(\alpha_k-1)/2-8) (1 + \underline{\gamma}-\varepsilon) + 3 k \alpha_k + n - 23 \\
& \geqslant k (3\alpha_k + (1 + \underline{\gamma}-\varepsilon)(\alpha_k-1)/2) + n - 23 - 8 \times 6 \\
\end{align*}

\begin{align*}
n \gamma_n
& \geqslant (n/(\overline{\alpha}+2) - 3) (3\overline{\alpha}-3\varepsilon + (1 + \underline{\gamma}-\varepsilon)(\overline{\alpha}-\varepsilon-1)/2) + n - 71 \\
& \geqslant \left(6 \overline{\alpha}-6\varepsilon + (1 + \underline{\gamma}-\varepsilon)(\overline{\alpha}-\varepsilon-1) + 2 \overline{\alpha}+4 \right) n/(2\overline{\alpha}+4) - \\
& \quad\quad 3 (3\overline{\alpha} + (1 + \underline{\gamma})(\overline{\alpha}-1)/2) - 71 \\
& \geqslant \left(\overline{\alpha} \underline{\gamma} + 9 \overline{\alpha}-\underline{\gamma} + 3 - (\overline{\alpha} + \underline{\gamma} -\varepsilon)\varepsilon \right) n/(2\overline{\alpha}+4) - 152.
\end{align*}

Hence, by choosing $n$ arbitrarily large, then $\varepsilon$ arbitrarily small, Lemma~\ref{lem:8.4} follows.
\end{proof}

From Lemma~\ref{lem:8.4}, we derive the asymptotic evaluation of the sequence $(\alpha_n)$,
as announced at the beginning of Section~\ref{sec:good2}.

\begin{proposition}\label{pro:8.5}
We have $\overline{\alpha} = \alpha_\infty = 2 + \sqrt{7}$.
\end{proposition}

\begin{proof}
Lemma~\ref{lem:8.4} states that $\underline{\gamma} \geqslant (\overline{\alpha} \underline{\gamma} + 9 \overline{\alpha} - \underline{\gamma} + 3)/(2\overline{\alpha}+4)$ or, equivalently,
that $\underline{\gamma} \geqslant (9\overline{\alpha}+3)/(\overline{\alpha}+5)$.
Since $\overline{\alpha} \geqslant \underline{\gamma}$, it follows that
$\overline{\alpha} \geqslant (9\overline{\alpha}+3)/(\overline{\alpha}+5)$, i.e., that
$\overline{\alpha} \geqslant 2 + \sqrt{7}$ or that $\overline{\alpha} \leqslant 2 - \sqrt{7} < 0$.
The latter case is obviously impossible, hence $\overline{\alpha} = \alpha_\infty = 2 + \sqrt{7}$.
\end{proof}

Finally, we may prove that the constant $\Delta$ of Theorem~\ref{thm:good-size} is optimal, as a consequence of the following result.

\begin{lemma}\label{lem:9.3}
For all real numbers $\Lambda > \Delta$, there exists a positive real number $K_{\Lambda}$ such that, for all $h \geqslant 1$,
there is a stack $(r_1,\ldots,r_h)$ for which $r_h \leqslant K_{\Lambda} \Lambda^h$.
\end{lemma}

\begin{proof}
Let $\varepsilon$ be an arbitrarily small real number such that $0 < \varepsilon < \Lambda/\Delta - 1$,
and let $N_\varepsilon$ be a large enough integer such that $\alpha_\ell \geqslant \alpha_\infty - \varepsilon$ and
for all $\ell \geqslant \lfloor (N-6)/(\alpha_\infty+2) \rfloor$.
Then, consider some integer $n_0 \geqslant N_\varepsilon$, and let $k = \lfloor (n-6)/(\alpha_\infty+2) \rfloor$.
As shown in the proof of Lemma~\ref{lem:8.4}, there are integers $x,y \in \{0,1,2,3\}$, and $n_1 = \lfloor (k(\alpha_k-1)-7-x-y) / 2 \rfloor$,
such that there exists a stack $(n_1-4,n_1,k(\alpha_k-1)-7-x-y, k\alpha_k-3-x, k(1+\alpha_k), n_0-4, n_0)$.
Since $\alpha_\infty \leqslant 5$, we have then
\begin{align*}
  \Delta^5 n_1 
  & \geqslant \Delta^5 (k (\alpha_\infty-1-\varepsilon)/2 - 8) \geqslant \Delta^5 \left((\alpha_\infty-1-\varepsilon)/(2\alpha_\infty+4) n_0 - 16\right) \\
  & \geqslant (\alpha_\infty-1-\varepsilon)/(\alpha_\infty-1) n_0 - 16 \Delta^5.
\end{align*}
It follows that $n_0 \leqslant \Delta^5 (n_1+16) (\alpha_\infty-1)/(\alpha_\infty-1-\varepsilon) \leqslant \Delta^5 (1+\varepsilon)^2 n_1 \leqslant \Lambda^5 n_1$.

Then, we repeat this construction, but replacing $n_0$ by $n_1$, thereby constructing a new integer $n_2$,
then replacing $n_0$ by $n_2$, and so on, until we construct an integer $n_\ell$ such that $n_\ell < N_\varepsilon$.
Doing so, we built a stack $(n_\ell-4,n_\ell,\ldots,n_0-4,n_0)$ of size $5\ell+2$,
and we also have $\Lambda^{5\ell} N_\varepsilon \geqslant \Lambda^{5\ell} n_\ell \geqslant \Lambda^{5(\ell-1)} n_{\ell-1} \geqslant \ldots \geqslant n_0$.
Choosing $K_\Lambda = N_\varepsilon$ completes the proof.
\end{proof}

\begin{proof}[Proof of Theorem~\ref{thm:good-size}]
We focus here on proving the second part of Theorem~\ref{thm:good-size}.
Consider a real constant $\Delta' > \Delta$, and let $\Lambda = (\Delta + \Delta') / 2$.
If $h$ is large enough, then
$h K_{\Lambda} \Lambda^h \leqslant (\Delta')^{h - 4}$.
Then, let $(r_1,\ldots,r_h)$ be a stack such that $r_h \leqslant K_{\Lambda} \Lambda^h$,
and let $n$ be some integer such that $(\Delta')^{h - 4} \leqslant n < (\Delta')^{h - 3}$, if any.
Considering the stack $(r_1,\ldots,r_h+m)$, where $m = n - (r_1+r_2+\ldots+r_h)$,
we deduce that $h_{\max} \geqslant h > 3 + \log_{\Delta'}(n)$.
Therefore, if $n$ is large enough, we complete the proof by choosing
$h = \lfloor \log_{\Delta'}(n) \rfloor + 3$.
\end{proof}

\end{document}